\documentclass[12pt]{article}

\usepackage{amsmath,amsthm,amssymb,thm-restate}
\usepackage{tikz,xspace,nicefrac,enumitem,caption}
\usepackage[normalem]{ulem}
\usepackage[margin=1in]{geometry}
\allowdisplaybreaks

\makeatletter
\g@addto@macro\normalsize{%
  \setlength\abovedisplayskip{5pt plus 3pt minus 3pt}
  \setlength\belowdisplayskip{5pt plus 3pt minus 3pt}
  \setlength\abovedisplayshortskip{5pt plus 3pt minus 2pt}
  \setlength\belowdisplayshortskip{5pt plus 3pt minus 2pt}
}
\makeatother

\date{15 May 2020}

\def\dfrac#1#2{\lower0.15ex\hbox{\large$\frac{#1}{#2}$}}

\numberwithin{equation}{section}

\def\({\bigl(}
\def\){\bigr)}

\newtheorem{thm}{Theorem}
\newtheorem{corollary}{Corollary}

\newtheorem{lemma}{Lemma}
\newtheorem{obs}{Observation}

\let\epsilon=\varepsilon

\newcommand{\cA}{\mathcal{A}}                     
\newcommand{\cC}{\mathcal{C}}                     
\newcommand{\cE}{\mathcal{E}}                     
\newcommand{\cH}{\mathcal{H}}                     
\newcommand{\cI}{\mathcal{I}}                     

\newcommand{\cP}{\mathcal{P}}                     
\newcommand{\cS}{\mathcal{S}}                     
\newcommand{\IN}{\mathbb{N}}                      
\newcommand{\nb}{\mathrm{N}}                      

\newcommand{\es}{\varnothing}                    

\newcommand{\sm}{\setminus}

\let\originalleft\left
\let\originalright\right
\renewcommand{\left}{\mathopen{}\mathclose\bgroup\originalleft}
\renewcommand{\right}{\aftergroup\egroup\originalright}

\newcommand{\wlg}{without loss of generality\xspace}
\newcommand{\Gv}[1][v]{G\sm\nb[#1]}
\newcommand{\GVN}[1]{G[V_{#1}\sm\nb[v_{#1}]]}

  \definecolor{darkred}{rgb}{0.6,0.0,0.0}
  \definecolor{lightred}{rgb}{1.0 0.8 0.8}
  \definecolor{lightblue}{rgb}{0.8 0.8 1.0}
  \definecolor{darkgreen}{rgb}{0.0 0.5 0.0}
  \definecolor{lightgreen}{rgb}{0.7 1.0 0.7}
  \definecolor{indigo}{rgb}{0.3 0.0 0.5}

\tikzset{ b/.style = { circle 
                     , draw
                     , thick
                     , inner sep = 0pt
                     , fill = black
                     , minimum size = 3.5pt
                     }
        , sb/.style = { circle 
                     , draw
                     , thick
                     , inner sep = 0pt
                     , fill = black
                     , minimum size = 2pt
                     }
        , w/.style = { circle 
                     , draw
                     , thick
                     , inner sep = 0pt
                     , fill = white
                     , minimum size = 3.7pt
                     }
        , sw/.style = { circle 
                     , draw
                     , thick
                     , inner sep = 0pt
                     , fill = white
                     , minimum size = 2pt
                     }
        , g/.style = { circle 
                     , draw
                     , thick
                     , inner sep = 0pt
                     , fill = lightgray
                     , minimum size = 4.0pt
                     }
        , i/.style = { circle 
                     , thick
                     , inner sep = 0pt
                     , minimum size = 4.0pt
                     }
        , R/.style = { circle 
                     , draw
                     , thick
                     , inner sep = 0pt
                     , fill = lightred
                     , minimum size = 5.5mm
                     }
        , B/.style = { circle 
                     , draw
                     , thick
                     , inner sep = 0pt
                     , fill = lightblue
                     , minimum size = 5.5mm
                     }
        , G/.style = { circle 
                     , draw
                     , thick
                     , inner sep = 0pt
                     , fill = lightgreen
                     , minimum size = 5.5mm
                     }
        , U/.style = { circle 
                     , draw
                     , thick
                     , inner sep = 0pt
                     , minimum size = 5.5mm
                     }
        }

\title{Counting weighted independent\\ sets beyond the permanent}
\author{
Martin Dyer\thanks{Work supported by
    EPSRC grants EP/S016562/1 and EP/S016694/1,``Sampling in hereditary classes''.}\\
\small School of Computing\\[-0.5ex]
\small University of Leeds\\[-0.5ex]
\small Leeds LS2~9JT, UK\\[-0.5ex]
\small\texttt{m.e.dyer@leeds.ac.uk}\\
\and Mark Jerrum\footnotemark[1]\\
\small School of Mathematical Sciences\\[-0.5ex]
\small Queen Mary University of London\\[-0.5ex]
\small Mile End Road, London E1 4NS, UK\\[-0.5ex]
\small\texttt{m.jerrum@qmul.ac.uk}
\and Haiko M\"{u}ller\footnotemark[1]\\
\small School of Computing\\[-0.5ex]
\small University of Leeds\\[-0.5ex]
\small Leeds LS2~9JT, UK\\[-0.5ex]
\small\texttt{h.muller@leeds.ac.uk}
\and Kristina Vu\v{s}kovi\'{c}\\
\small School of Computing\\[-0.5ex]
\small University of Leeds\\[-0.5ex]
\small Leeds LS2~9JT, UK\\[-0.5ex]
\small\texttt{k.vuskovic@leeds.ac.uk}
}

\begin{document}

\maketitle

\begin{abstract}
Jerrum, Sinclair and Vigoda (2004) showed that the permanent of any square matrix can be estimated in polynomial time.
This computation can be viewed as approximating the partition function of edge-weighted matchings in a bipartite graph.
Equivalently, this may be viewed as approximating the partition function of vertex-weighted independent
sets in the line graph of a bipartite graph.\\
\indent Line graphs of bipartite graphs are perfect graphs, and are known to be precisely the class of
(claw,\,diamond,\,odd hole)-free graphs.
So how far does the result of Jerrum, Sinclair and Vigoda extend? We first show that it extends to
(claw,\,odd hole)-free graphs, and then show that it extends to the even larger class of (fork,\,odd hole)-free graphs.
Our techniques are based on graph decompositions, which have been the focus of much recent work
in structural graph theory, and on structural results of  Chv\'atal and Sbihi (1988), Maffray and Reed (1999)
and Lozin and Milani\v{c} (2008).
\end{abstract}

\section{Introduction}\label{intro}

Independent sets are central objects of study in graph theory.\footnote{Definitions not given here appear in~Section~\ref{sec:prelim} below.}
 In general, finding a largest independent set is a very hard problem.
Indeed, it is known to be hard to approximate the size of this set within a ratio $n^{1-\epsilon}$ for graphs on $n$ vertices and any $\epsilon>0$, unless $\mathsf{P}=\mathsf{NP}$. This has led to an emphasis on studying this problem in particular classes of graphs.

In particular, there has been a focus on \emph{hereditary} classes, that is, classes that are closed under vertex deletions. Equivalently, such a class can be defined by a (not necessarily finite) set of \emph{forbidden (induced) subgraphs}. 

\emph{Matchings} are a particular case. A matching in a graph $G=(V,E)$ is an independent set in the \emph{line graph} of the root graph $G$. (``Graph'' will mean ``simple undirected graph'', unless otherwise stated.) Edmonds~\cite{edmonds} showed that a maximum weighted matching in any graph can be found in polynomial time. Beineke~\cite{beineke} showed that line graphs can be characterised by nine forbidden subgraphs. Of these, the \emph{claw} (see Fig.~\ref{fig:clawdiagem}) seems the most important for algorithmic questions. Thus Minty~\cite{Minty} extended Edmonds' algorithm to the  larger class of \emph{claw-free} graphs. 

Claw-free graphs have been studied extensively by several authors, including Chudnovsky and Seymour in a long sequence of papers culminating in~\cite{ChuSey}. These papers give a decomposition of claw-free graphs, which unfortunately is non-algorithmic. However, this has been simplified and extended in \cite{FOS,NobSas} to give an efficient decomposition which supports finding a maximum weighted independent set.

In this paper we are concerned with counting problems, and for these it is important to distinguish weighted and unweighted (or unary weighted) variants, even more so than with optimisation problems.  For example, there is an efficient approximation algorithm for counting unweighted matchings in a general graph, but the existence of an approximation algorithm for counting weighted matchings remains an open question.  It is weighted counting problems that we focus on here. Hereditary classes are particularly suitable for counting problems, since they are self-reducible by vertex deletion.

Since claw-free graphs include line graphs, the \textsf{\#P}-completeness result of Valiant~\cite{Valiant} for matchings implies that \emph{exact counting} of independent sets in polynomial time is unlikely. Even polynomial time \emph{approximate} counting remains an open question for general line graphs in the weighted setting.  However, building on an earlier pseudopolynomial algorithm of Jerrum and Sinclair~\cite{JS}, Jerrum, Sinclair and Vigoda~\cite{JSV} made an important breakthrough in approximate counting. They showed that approximately counting weighted \emph{perfect matchings} in a \emph{bipartite} graph is in polynomial time, the \emph{permanent} approximation problem.

Our goal is to extend the result of~\cite{JSV} to larger classes of graphs. It might be expected that
the right direction for this would be to matchings in general graphs but, as noted above, this remains an open problem, and a positive solution seems increasingly unlikely.
An important requirement of the proof of~\cite{JSV} is that the graph should have no odd cycles, which places it precisely in the class of bipartite graphs. Indeed, \v{S}tefankovi\v{c}, Vigoda and Wilmes~\cite{SVW} have given
a family of nonbipartite graphs for which the algorithm of~\cite{JSV} does not run in polynomial time. Interestingly, from the viewpoint of this paper, they also show that weighted matchings in these graphs, and in a more general class of graphs that are ``close to bipartite'', can be counted in polynomial time, using the algorithm of~\cite{JSV} with a graph decomposition technique. In~\cite{SVW}, this is the Gallai-Edmonds decomposition.

Here we take a different direction to generalise~\cite{JSV}, regarding approximating the permanent as the problem of approximately counting weighted independent sets in line graphs of bipartite graphs. We show that these two problems are  polynomial time equivalent. That approximating the permanent is reducible to approximately counting weighted independent sets is shown in Section~\ref{sec:claw W_k}, and that counting arbitrarily weighted independent sets in line graphs of bipartite graphs is reducible to approximating the permanent is shown in Section~\ref{sec:permanent}.

An important property of line graphs of bipartite graphs is that they are \emph{perfect}. So it might be hoped that the appropriate generalisation of the result of~\cite{JSV} would be to counting independent sets in perfect graphs. That this class can be recognised in polynomial time was shown by Chudnovsky, Cornu{\'{e}}jols, Liu, Seymour and Vu\v{s}kovi\'{c}~\cite{ChCLSV}. The maximum independent set in a perfect graph can also be found in polynomial time, using a convex optimisation algorithm of  Gr\"otschel, Lov\'asz, and Schrijver~\cite{GLS}, though no combinatorial algorithm is yet known for this problem.

However, approximately counting independent sets in perfect graphs appears intractable in general. Bipartite graphs are perfect, but approximately counting independent sets in bipartite graphs defines the complexity class \textsf{\#BIS}. This class was introduced by Dyer, Goldberg, Greenhill and Jerrum~\cite{DyGoGJ}, and hardness for the class has since been used as evidence for the intractability of various approximate counting problems.

Trotter~\cite{Trot} suggested the smaller class of \emph{line-perfect} graphs. These are graphs whose line graph is perfect. Trotter showed that a graph is line-perfect if and only if it contains no odd cycle of size larger than three. Independent sets in the line graph (matchings in the root graph) appear a natural target for generalising~\cite{JSV} but, in fact, they are  a proper subclass of those that we will consider here.

Line graphs of bipartite graphs have a simple set of forbidden subgraphs \cite{HaHo74}.
These are the claw, all odd holes and the \emph{diamond} (see Fig.~\ref{fig:clawdiagem}). See Maffray and Reed~\cite[Thm.\,4]{MaRe}, who also gave the corresponding result~\cite[Thm.\,5]{MaRe} for line graphs of bipartite \emph{multigraphs}. These have the claw, \emph{gem}, \emph{4-wheel} and
odd holes (see Fig.~\ref{fig:clawdiagem}) as forbidden subgraphs. Of these, excluding the claw
and the odd holes appears important in extending the algorithm of~\cite{JSV}.
This results from the ``canonical paths'' argument used in its proof.
However, the diamond, gem or 4-wheel do not appear important in this respect.

Here we establish this claim. We extend the result of~\cite{JSV} to the class of graphs which excludes only claws and odd holes. This is essentially the class of \emph{claw-free perfect} graphs. (See Section~\ref{structure} below.) These form a main focus of this paper, and we show that the algorithm of~\cite{JSV} can be extended to approximate the total weight of independent sets for graphs in this class.
The structure of graphs in this class was characterised by Chv\'atal and Sbihi~\cite{ChSb} and Maffray and Reed~\cite{MaRe}. They gave a polynomial time \emph{decomposition} algorithm that splits the graph into simpler parts. We use their results to show that the algorithm of~\cite{JSV} can be applied directly to count weighted independent sets in (claw,\,odd hole)-free graphs, a slightly larger class than
claw-free perfect graphs. Since line graphs of bipartite graphs are a proper subclass of (claw,\,odd hole)-graphs,
this is a natural generalisation of the  result of~\cite{JSV}.

Chudnovsky and Plumettaz~\cite{ChPl} have given a different decomposition of claw-free perfect graphs, which has the additional property of \emph{composability}. That is, the rules used to decompose the graph can be applied in reverse to create precisely the graphs in the class (and no more). Unfortunately, this results in a considerably more complex decomposition than that of~\cite{ChSb,MaRe}. We make no use of this here, since we do not need composability. Moreover,~\cite{ChPl} does not give a polynomial time algorithm for its decomposition. The ideas in~\cite{FOS,NobSas} would give a polynomial time decomposition, but it is unclear whether this supports counting.

In Section~\ref{sec:clawoddholefree} we will develop a polynomial time algorithm for
approximately counting all weighted independent sets in a (claw,\,odd hole)-free graph,
generalising the algorithm of~\cite{JSV}.

Observe that the algorithm of \cite{DyGrMu} runs in polynomial time, and counts all weighted independent sets
in an arbitrary claw-free graph $G=(V,E)$ with  unary weights. (See also Matthews~\cite{Matthe08}.)
This generalises Jerrum and Sinclair's matching algorithm~\cite{JS}.

So what do we achieve by restricting to (claw,\,odd hole)-free graphs? The gain is that our algorithm is genuinely polynomial time, whereas that of \cite{DyGrMu} is only pseudopolynomial. In particular, this allows us to approximate the total weight of independent sets of any given size $k$. We can estimate the total weight of \emph{maximum} independent sets, which corresponds to counting maximum matchings in the root graph of a line graph.
We then further relax the conditions on the class. We cannot relax first the odd hole condition, since this would take us into more general claw-free graphs, and might require counting matchings in general graphs. Therefore we consider relaxing the claw-free condition.

Lozin and Milani\v{c}~\cite{LozMil} described a polynomial time algorithm for finding a maximum weight\-ed independent set in a \emph{fork-free} graph. That is, a graph with only the \emph{fork} (see Fig.~\ref{fig:clawdiagem}) as a forbidden subgraph. Clearly, this is a proper superclass of claw-free graphs, since the claw is a subgraph of the fork. In Section~\ref{sec:forkcount}, we show how our methods can be combined with ideas of~\cite{LozMil} to count arbitrarily weighted independent sets in (fork,\,odd hole)-free graphs. Again, these are a proper superclass of (claw,\,odd hole)-free graphs. This gives a further nontrivial generalisation of the result of~\cite{JSV}.
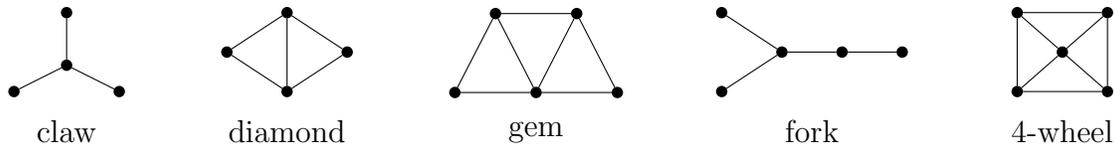
\begin{figure}
\begin{center}
  \begin{tikzpicture}[scale=0.35]
    \node[b] (v4) at (2,0) {};
    \node[b] (vv) at (4,1) {};
    \node[b] (v0) at (4,3) {};
    \node[b] (v2) at (6,0) {};
    \draw (v4)--(vv)--(v0)  (vv)--(v2);
    \node at (4,-1.5) {claw};
  \end{tikzpicture}\hspace{1.1cm}
  \begin{tikzpicture}[xscale=0.4,yscale=0.35]
    \node[b] (v1) at (0,1.5) {};
    \node[b] (v2) at (2,3) {};
    \node[b] (v3) at (4,1.5) {};
    \node[b] (v4) at (2,0) {};
    \draw (v1)--(v2)--(v3)--(v4)--(v1) (v2)--(v4) ;
    \node at (2,-1.5) {diamond};
  \end{tikzpicture}\hspace{1.1cm}
  \begin{tikzpicture}[xscale=0.27,yscale=0.35]
    \node[b] (v1) at (0,0) {};
    \node[b] (v2) at (2,3) {};
    \node[b] (v3) at (4,0) {};
    \node[b] (v4) at (6,3) {};
    \node[b] (v5) at (8,0) {};
    \draw (v1)--(v2)--(v3)--(v4)--(v5) (v1)--(v3)--(v5) (v2)--(v4) ;
    \node at (4,-1.5) {gem};
  \end{tikzpicture}\hspace{1.1cm}
  \begin{tikzpicture}[xscale=0.4,yscale=0.35]
    \node[b] (v1) at (0,3) {};
    \node[b] (v2) at (2,1.5) {};
    \node[b] (v3) at (0,0) {};
    \node[b] (v4) at (4,1.5) {};
    \node[b] (v5) at (6,1.5) {};
    \draw (v1)--(v2)--(v3) (v2)--(v4)--(v5) ;
    \node at (3,-1.5) {fork};
  \end{tikzpicture}\hspace{1.1cm}
  \begin{tikzpicture}[xscale=0.4,yscale=0.35]
    \node[b] (v1) at (0,3) {};
    \node[b] (v2) at (3,3) {};
    \node[b] (v3) at (3,0) {};
    \node[b] (v4) at (0,0) {};
    \node[b] (v5) at (1.5,1.5) {};
    \draw (v1)--(v2)--(v3)--(v4)--(v1) (v1)--(v5)--(v2) (v3)--(v5)--(v4) ;
    \node at (1.5,-1.5) {4-wheel};
  \end{tikzpicture}
\end{center}\caption{Claw, diamond, gem, fork and 4-wheel}\label{fig:clawdiagem}
\end{figure}

\subsection{Preliminaries}\label{sec:prelim}

Let $\IN=\{1,2,\ldots\}$ denote the natural numbers, and $\IN_0=\IN\cup\{0\}$. If $n\in\IN$, let $[n]=\{1,2,\ldots,n\}$.
For a set $S$, $S^{(2)}$ will denote the set of subsets of $S$ of size exactly 2. 

Throughout this paper, graphs are always simple and undirected. Let $G=(V,E)$ be a graph. 
 We denote its vertex set by $V(G)$, and its edge set by $E(G)\subseteq V^{(2)}$. We write an edge $e\in E$ between $v$ and $w$ in $G$ as $e=vw$, or $e=\{v,w\}$ if the $vw$ notation is ambiguous. If $U,W\subseteq V$ with $U\cap W=\es$, we will denote the $U,W$ \emph{cut} by $(U,W)=\{uw \in E : u\in U,w\in W\}$. We also consider \emph{multigraphs}, in which $E$ may have \emph{parallel edges}. That is, $E$ is a \emph{multiset} with elements in $V^{(2)}$.

For a graph $G=(V,E)$, we will write $n=|V|$ and $m=|E|$, unless stated otherwise. The \emph{empty graph} $G=(\es,\es)$ is the unique graph with $n=0$.  Also, $G=(V,V^{(2)})$, is the complete graph on $n$ vertices. The \emph{complement} of any graph $G=(V,E)$ is $\overline{G}=(V,V^{(2)} \sm E)$.

The \emph{neighbourhood} of $v\in V$ will be denoted $\nb(v)$, and $\nb[v]=\nb(v)\cup\{v\}$.
Then the \emph{degree} $\deg(v)$ of $v$ is $|\nb(v)|$.
More generally, the neighbourhood of a set $U\subseteq V$, is $\nb(U)=\{v\in V\sm U: uv\in E\textrm{ for some }u\in U\}$, and $\nb[U]= U\cup\nb(U)$.
We will say that a vertex $v\in V \sm U$ is \emph{complete} to $U$ if $U\subseteq \nb(v)$, and
\emph{anticomplete} if $U\cap\nb(v)=\es$. More generally, a set $U\subseteq V$ is complete to a set $W\subseteq V \sm U$ if every $u\in U$ is adjacent to every $w \in W$, and anticomplete if $(U,W)=\es$. Observe that this is a symmetric relation between $U$ and $W$. The graph $G=(V,E)$ is \emph{connected} if $V$ cannot be partitioned into sets $U,W$ that are anticomplete.

The term ``induced'' subgraph will always mean a vertex-induced subgraph.
If $U\subseteq V$, we will write $G[U]$ for the subgraph of $G$ induced by $U$.
Where ``subgraph'' is used without qualification, it will always mean induced subgraph.
Then a class $\cC$ of graphs is called \emph{hereditary} if $G[U]\in\cC$ for all $G\in\cC$ and $U\subseteq V$.
If $U\subseteq V$, we will often write $G\sm U$ as shorthand for $G[V\sm U]$.

We say a graph $G$ \emph{contains} a graph $H$ if it has an induced subgraph isomorphic to $H$,
and  $H$ is a \emph{forbidden subgraph} for the graph class $\cC$ if no graph in $\cC$ contains $H$.
It is easy to see that any hereditary class can be characterised by a (possibly infinite)
set $\cH$ of minimal forbidden subgraphs. In this case we refer to $\cC$ as the class of $\cH$-free graphs.

An \emph{odd hole} in a graph $G$ is a subset $H \subseteq V$, with $|H|\geq 5$ and odd, such that $G[H]$ is a simple cycle. A \emph{perfect} graph $G$ is such that neither $G$ nor its complement $\overline{G}$ contains an odd hole. A hole in $\overline{G}$ is called an \emph{antihole} in $G$. Perfect graphs were originally defined differently, but the equivalence to this definition was proved in~\cite{ChRoST}.

The \emph{line graph} $L(G)=(E,\cE)$ of a multigraph $G=(V,E)$ has $\cE=\big\{ \{xy,yz\}: xy,yz\in E\big\}$.
We will write $G=L^{-1}(G')$ for the inverse operation, when it is defined, and call $G$ the \emph{root} multigraph of the line graph $G'$. Note that $G=L^{-1}(G')$ is not unique when $G$ is a multigraph, whereas it is unique for $|V|>4$ when $G$ is a graph. When $L^{-1}(G')$ is not uniquely defined, we may choose it to be any multigraph $G$ such that $G'=L(G)$.

A set $S \subseteq V$ is \emph{independent} (or \emph{stable}) in $G$ if
$G[S]$ is edgeless. The empty set $\es$ is an independent set in every graph.
By $\cI(G)$ we denote the set of all independent sets of
$G$, and $\cI_k(G) = \cI(G) \cap V^{(k)}$ for $k \in \IN_0$.
The largest $k$ for which $\cI_k(G)\neq\es$ is the \emph{independence number}
$\alpha(G)$ of $G$.

For further information on graph theory, see~\cite{survey,diestel}, for example.

We will suppose that the vertices $v \in V(G)$ are equipped with non-negative \emph{weights} $w(v)\in\mathbb{R}$.
We will denote such a vertex-weighted graph by $(G,w)$, or simply $G$ when the vertex weights $w$ are understood. Two weighted graphs $(G_1,w_1)$, $(G_2,w_2)$ will be called \emph{isomorphic} if $G_1,G_2$ are isomorphic as graphs, though we may have $w_1\neq w_2$.
Since we are considering only approximation, we may assume here that $w(v)\in\mathbb{Q}$.
The weight of a subset $S$ of $V$ is then defined to be $w(S)=\prod_{v \in S} w(v)$.\footnote{Note the difference from the corresponding definition $\sum_{v \in S} w(v)$ used in optimisation.}
Then let $W_k(G) = \sum_{S \in \cI_k(G)} w(S)$,
and $W(G) = \sum_{S \in \cI(G)} w(S)=\sum_{k=0}^{\alpha(G)} W_k(G)$.
In particular, we have $W_0(G)=1$, $W_1(G)=\sum_{v \in V} w(v)$ and $W_2(G)=\sum_{uv \notin E} w(u)w(v)$.

We will use only the following simple properties of $W(G)$.
If $G$ has connected components $C_1,C_2,\ldots,C_r$, then $W(G)=\prod_{i=1}^{r}W(C_i)$,
and if $\cS_1,\cS_2,\ldots,\cS_s$ partitions $\cI(G)$, then $W(G)=\sum_{i=1}^{s} \sum_{I\in\cS_i}w(I)$.

We will say that a vertex-weighted graph $(G',w')$ is \emph{equivalent} to a vertex-weighted graph $(G,w)$ if $W_k(G)=W_k(G')$, for all $k\in\IN_0$, and hence $W(G)=W(G')$. In particular, this implies $\alpha(G)=\alpha(G')$. Observe that equivalent weighted graphs are not necessarily isomorphic,
and isomorphic weighted graphs are not necessarily equivalent.

Note that, if $w(v)=0$ for any $v\in V$, then $G$ is equivalent to $G[V\sm\{v\}]$, thus we can consider such vertices as present in or deleted from $G$, whichever is more convenient. We will assume that such vertices are deleted before carrying out computations, so we may assume that $w(v)>0$ for all $v\in V$.

If $w(v)=1$ for all $v\in V$, then $W_k(G)=|\cI_k(G)|$, the number of independent sets of size $k$ in $G$,
and $W(G)=|\cI(G)|$ counts all independent sets in $G$. However, we also refer to the case with non-unit weights
as ``counting''.

A central theme of structural graph theory has been \emph{decomposition}, that is, breaking a graph into smaller
pieces that have stronger properties than the original, such that the pieces are all connected to each other
in some canonical fashion.
Our counting algorithms  for (claw,\,odd hole)-free and (fork,\,odd hole)-free graphs are based on two graph
decompositions, clique cutset decomposition and modular decomposition, respectively.
We will describe these in Section~\ref{sec:clawoddholefree} and Section~\ref{sec:forkcount} respectively.

We consider approximating $W(G)$ and $W_k(G)$ in the following sense. An FPRAS (fully polynomial randomized
approximation scheme) is an algorithm which produces an estimate $\widehat{W}$ of a quantity $W$ such that
\[ \Pr\big((1-\varepsilon)W \leq \widehat{W} \leq (1+\varepsilon)W \big) \, \geq \, \nicefrac34\,. \]
The key FPRAS we employ here is that of Jerrum, Sinclair and Vigoda~\cite{JSV} for the \emph{permanent}.
This uses the Markov chain approach to approximate counting, but we will not need the interior details of the algorithm.
Essentially, we use~\cite{JSV} as a ``black box'' here. We note the equivalence of approximate counting
with approximate random generation~\cite{JVV}, but we make no direct use of this here.

For further information on approximate counting, see~\cite{jerrumbook}, for example.

\section{Approximating $W(G)$ in (claw,\,odd hole)-free graphs}\label{sec:clawoddholefree}
We develop an algorithm for approximating $W(G)$ in (claw,\,odd hole)-free graph using clique cutset decomposition.
But first we describe this method, and its application to counting, in a general setting.

\subsection{Clique cutset decomposition}\label{sec:cliques}

A \emph{clique} $K$ in a graph $G=(V,E)$ is a subset of $V$ such that $G[K]$
is a complete graph.
$K \subseteq V$ is a \emph{clique cutset} of $G$ if $K$ is a clique and $G \sm K$ is disconnected. In particular, $\es$ is a clique cutset of every disconnected graph. For a clique cutset $K$ of $G$, let $(A,B)$ be a partition of $V \sm K$ such that $A$ is anticomplete to $B$. The subgraphs $G[A \cup K]$ and $G[B \cup K]$ are \emph{blocks} of the decomposition of $G$ by $K$.
These blocks may themselves have clique cutsets, so may contain further blocks. A block with no clique cutset is called an \emph{atom}. The decomposition can be presented in the form of a tree, in which the interior vertices are cliques, and the leaves are atoms. Tarjan~\cite{Tarjan} gave $O(mn)$ algorithm for a particular tree representation. This gives a binary decomposition tree in which all the interior nodes (cliques) form a path. If the tree has height $h$, we will number the atoms $A_0,A_1,\ldots,A_h$
and cliques $K_1,K_2,\ldots,K_h$ from the bottom up in this tree. While the atoms are all different, a clique can occur several times in a decomposition tree. See Fig.~\ref{fig:cliquedecomposition}, and see~\cite{Tarjan} for further information. We describe how this decomposition may be used for computing $W(G)$
in Section~\ref{sec:cutsetcounting}.

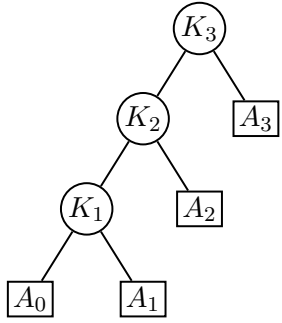
\begin{figure}
\begin{center}
\begin{tikzpicture}[line width=0.75pt,minimum size=0mm,xscale=1.5,yscale=1.2,font=\small]
\draw (0,0) node[rectangle,draw,inner sep=2pt] (A0) {$A_0$} (1,0) node[rectangle,draw,inner sep=2pt] (A1) {$A_1$}
(0.5,1) node[circle,draw,inner sep=1pt] (C1) {$K_1$} (1.5,1) node[rectangle,draw,inner sep=2pt] (A2) {$A_2$}
(1,2) node[circle,draw,inner sep=1pt] (C2) {$K_2$} (2,2) node[rectangle,draw,inner sep=2pt] (A3) {$A_3$}
(1.5,3) node[circle,draw,inner sep=1pt] (C3) {$K_3$} ;
\draw (A0)--(C1)--(C2)--(C3) (C1)--(A1) (C2)--(A2) (C3)--(A3) ;
\end{tikzpicture}
\end{center}
\caption{Clique decomposition tree}\label{fig:cliquedecomposition}
\end{figure}

\subsection{Approximating $W(G)$ using clique cutset decomposition}\label{sec:cutsetcounting}
Let $\cC$ be a hereditary class of graphs such that all graphs in $\cC$ have a clique cutset decomposition
with all atoms in some hereditary class $\cA\subset\cC$, and we can approximate $W(G)$ for any weighted
$G=(V,E)$ in $\cA$ in time $T_{\cA}(n)=\Omega(n)$, where $T$ is assumed convex.  We show how to determine $W(G)$ for the entire graph $G$
in time  $T_{\cC}(n)\leq 2nT_{\cA}(n)$.

The decomposition tree in Section~\ref{sec:cliques} has cliques $K_1,K_2,\dots,K_h$ and atoms $A_0,A_1,\ldots,A_h$, where $h\leq n$. The root of the tree is $K_h$. Let $A'_i=A_i\sm K_i$, $s_i=|K_i|$, $a'_i=|V(A'_i)|$. Let $G_i$ be the graph formed by deleting the vertices of $A'_h,A'_{h-1},\dots, A'_{i+1}$ from $G$.
Thus $G_h=G$, and $G_{i-1}=G_i\sm V(A'_i)$. We will account for the independent sets intersecting $A'_i$ by revising
the weights on the vertices in $K_i$, in a similar way to Tarjan's~\cite{Tarjan} approach to the maximum weight
independent set problem.

Since $K_i$ is a clique cutset in $G_i$, we may partition $\cI(G_i)$ by the value of $I\cap K_i$,
which is either $\{v\}$ $(v\in K_i)$ or $\es$. Then we partition $G_i$ into three vertex-disjoint
subgraphs $G_{i-1}\sm K_i$, $G[K_i]$, and $A'_i$. So we may write
\begin{align*}
 W(G_i)&= W(G_{i-1}\sm K_i)\,W(A'_i) + \sum_{v\in K_i}W(G_{i-1}\sm\nb[v])\,w(v)\,W(A'_i\sm \nb(v))\\
 &=W(A'_i)\,\Big(W(G_{i-1}\sm K_i) + \sum_{v\in K_i}W(G_{i-1}\sm\nb[v])\,w(v)\,W(A'_i\sm \nb(v))/W(A'_i)
 \Big) \\
 &= W(A'_i)\,W(G_{i-1})\;,
 \intertext{where the vertex weights in $G_{i-1}$ relate to those in $G_i$ by}
 w(v)& \gets w(v)\,W(A'_i\sm \nb(v))/W(A'_i)\quad (v\in K_i),\qquad w(v) \gets w(v)\quad (v\notin K_i)\;.
\end{align*}
Thus, since $G_h=G$, we may compute $W(G)$ by induction as
\[ W(G)=W(A_0)\prod_{i=1}^hW(A'_i)\;,\]
where we update the vertex weights as above at each stage. Note that the weights of some vertices may change several times in this process, since the cliques of Tarjan's decomposition are not necessarily vertex-disjoint.

At stage $i$, we have to perform $s_i+1$ computations on subgraphs of $A'_i$, which are all in $\cA$. Thus the total time is
$T_{\cC}(n)=\sum_{i=1}^h (s_i+1)T_{\cA}(a'_i) \leq 2nT_{\cA}(n)$, since $h\leq n$, $s_i\leq n$, $\sum_{i=1}^h a'_i\leq n$
and $T$ is convex.

Note that this analysis deals only with applications of the algorithm for $\cA$. It ignores the effect of the bit-size of the vertex weights on $T_{\cA}(n)$. This distinction is not so important for optimisation, but is much more important for counting, since contracting modules (see \ref{sec:moddecomp}) can cause exponential growth in the weights. The same comment applies to the algorithm of Section~\ref{sec:forkcount}. However, we do not pursue this issue further in this paper.

\subsubsection{Error Analysis}\label{sec:clawerror}
We are only approximating the weight of graphs in $\cA$, so we must show that the resulting error in $W(G)$ can be controlled for $G\in\cC$.

Suppose we approximate to a factor $(1\pm\varepsilon/n^2)$ throughout.
Then, by induction, the weights in $A_{h-i}$ will have relative error at most $(1\pm\varepsilon/n^2)^i$. Thus the estimate of the total weight of $A_0$ will have relative error at most $(1\pm\varepsilon/n^2)^h$. The error in $W(A'_{h-i})$ will be at most $(1\pm\varepsilon/n^2)^i$, so the error in $\prod_{i=1}^hW(A'_i)$ is at most $(1\pm\varepsilon/n^2)^{h(h-1)/2}$. Hence the error in $W(G)$ is at most $(1\pm\varepsilon/n^2)^{h(h+1)/2}$. Since $h < n-1$, the error is a most $(1\pm\varepsilon/n^2)^{n^2/2}$, which is at most $(1\pm\varepsilon)$ for $\varepsilon<1$.
So the overall error can be kept within any desired relative error $\varepsilon$ by performing the weight estimations for all  graphs in $\cA$ to within error $\varepsilon/n^2$.

\subsection{Structure of claw-free perfect graphs}\label{structure}

In our application of the method of~Section~\ref{sec:cutsetcounting}, $\cC$ will be
the class of (claw,\,odd hole)-free graphs, and $\cA$ will be a class of graphs
that we will define below. We must examine approximate counting in this class,
but first we review the structural results which allow us to apply clique cutset decomposition.

Chv\'atal and Sbihi~\cite{ChSb} investigated the structure of claw-free \emph{perfect} graphs as a special class of perfect graphs. These are closely related to (claw,\,odd hole)-free graphs. The difference is that odd antiholes are also forbidden. The following lemma of Ben Rebea explains that relation.
\begin{lemma}[Ben Rebea] \label{lem:benrebea}
Let $G$ be a connected claw-free graph with
$\alpha(G)\geq3$. If $G$ contains an odd antihole then it contains a hole of length five.\qed
\end{lemma}
\begin{corollary}\label{cor:benrebea}
A claw-free graph with $\alpha(G)\geq3$ is perfect if and only if it has no odd hole.
\end{corollary}
\begin{proof}
From the Strong Perfect Graph Theorem~\cite{ChRoST}, Lemma~\ref{lem:benrebea} implies that a claw-free graph with $\alpha(G)\geq 3$ is perfect if and only if it contains no odd hole.
\end{proof}
Chv\'atal and Sbihi~\cite{ChSb} gave a decomposition theorem via clique cutsets for claw-free perfect graphs.
As described in Section~\ref{sec:cliques},
a clique cutset decomposition can be described by a binary tree whose interior vertices are cliques, and whose leaves are atoms.
\begin{thm}[Chv\'atal and Sbihi]\label{thm:Chvatal}
If a claw-free perfect graph has no clique cutset then it is either elementary or peculiar. \qed
\end{thm}
We will describe elementary and peculiar graphs below. These will be the atoms of the decomposition.

\subsubsection{Peculiar graphs}
A \emph{peculiar} graph is constructed as follows. A set $K$ of vertices, initially a clique, is partitioned into six non-empty subsets $A_1, A_2,  A_3, B_1, B_2, B_3$. At least one edge is removed from each of the edge sets $(A_1,B_2)$, $(A_2,B_3)$ and $(A_3,B_1)$. Finally, three disjoint nonempty cliques $K_1, K_2, K_3$ are added, and  each vertex in $K_i$ is made adjacent to every vertex in $K \sm (A_i\cup B_i)$ for $i=1,2,3$.

The smallest peculiar graph, with $|A_i|,|B_i|,|K_i|=1$ ($i=1,2,3$) is shown in Fig.~\ref{fig:peculiar}. The black vertices are $A$'s, the white $B$'s and the grey $K$'s. This graph is a template for all peculiar graphs, as shown by Chv\'atal and Sbihi~\cite{ChSb}.

We will need the following simple observation about peculiar graphs.
\begin{lemma}\label{lem:peculiar}
A peculiar graph $G=(V,E)$ has independence number $\alpha(G)=3$. Any independent set of size three has one vertex in each of $K_1,K_2,K_3$.
\end{lemma}
\begin{proof}
Note that $K_1\cup A_3\cup B_2$, $K_2\cup A_1\cup B_3$, $K_3\cup A_2\cup B_1$ are three cliques which cover $V$, so  $\alpha(G)\leq 3$. However, we can form an independent set of size three by taking one vertex from each of $K_1,K_2,K_3$, so $\alpha(G)\geq3$.

Let $I=\{v_1,v_2,v_3\}$ be any maximum independent set in $G$. Suppose first $I\cap K_i=\es$ for all $i=1,2,3$. Then $I$ is contained in $K$.
But $K$ is a perfect graph with vertices contained in two disjoint cliques, $A_1\cup A_2\cup A_3$ and $B_1\cup B_2\cup B_3$. Thus $K$ is a cobipartite graph, with $\alpha(K) \leq 2$. Hence $\alpha(G) \leq 2$, a contradiction.

Thus, \wlg, assume $v_1\in K_1\cap I$.
Now $\nb[v] = (K \cup K_1)\sm(A_1\cup B_1)$. So $v_2,v_3\in A_1\cup B_1\cup K_2\cup K_3$. But $A_1\cup B_1\cup K_i$ is a clique for $i=2,3$. Thus, if $v_2\in A_1\cup B_1$, $v_3$ cannot exist. Thus $v_2\in K_2$, \wlg, and then we must have $v_3\in K_3$.
\end{proof}
In Fig.~\ref{fig:peculiar},  the three corner triangles cover all the vertices, and the three corner vertices form an independent set.

\begin{figure}
\begin{center}
\begin{tikzpicture}[line width=0.5pt,inner sep=0pt,minimum size=0mm,scale=1.5,font=\small]
 \foreach \x in {0,2,4}{\draw (60*\x:1cm) node[b] (\x) {} ; }
 \foreach \x in {1,3,5}{\draw (60*\x:1cm) node[w] (\x) {} ; }
 \foreach \x in {0,1,2}{\draw (120*\x-30:1.7cm) node[g] (k\x) {} ; }
 \draw (0)--(1)--(2)--(3)--(4)--(5)--(0) (0)--(2)--(4)--(0) (1)--(3)--(5)--(1);
 \draw (0)--(k0)--(5) (1)--(k1)--(2) (3)--(k2)--(4) ;
 \draw (4)edge[bend right=30](k0) (1)edge[bend left=30](k0) ;
 \draw (0)edge[bend right=30](k1) (3)edge[bend left=30](k1) ;
 \draw (2)edge[bend right=30](k2) (5)edge[bend left=30](k2) ;
 \end{tikzpicture}
 \caption{Minimal peculiar graph from \cite{ChSb}}\label{fig:peculiar}
\end{center}
\end{figure}
Peculiar graphs do not form an hereditary class. If in a subgraph $G$ of a peculiar graph, $K_i=\es$ holds for any $i=1,2,3$, it follows that $\alpha(G)\leq 2$. However, if $\alpha(G)\leq 2$, then $G$ is a clique or a cobipartite graph. But both of these are elementary graphs, as defined in Section~\ref{sec:elementary} below. Thus we may insist that a peculiar graph has $K_i\neq \es$ ($i=1,2,3$) and $\alpha(G)=3$.

\subsubsection{Elementary graphs}\label{sec:elementary}
Chv\'atal and Sbihi called a graph $G=(V,E)$ \emph{elementary} if $E$ can be two-(edge)-coloured so that edges $xy,yz\in E$ have distinct colours whenever $xz\notin E$. Such a colouring is called elementary. It is clear from this that elementary graphs form a hereditary class. Whether $G$ has an elementary colouring can be checked by forming the \emph{Gallai graph} Gal$(G)=(E,\cE)$, where  $\{xy,yz\}\in \cE$ if and only if $(x,y,z)$ is a $P_3$. Clearly Gal$(G)$ can be constructed in time $O(mn)$, by taking all pairs of $xy\in E$ and $z\in V$, and checking that $yz\in E$, $xz\notin E$. Then an elementary colouring exists if and only if Gal$(G)$ is bipartite. If so, the two colour classes of Gal$(G)$ give an elementary colouring of $G$.  This can be recognised in $O(|\cE|)$ time by breadth-first search, so this gives an $O(mn)$ algorithm for recognising elementary graphs and determining an elementary colouring. Note that some edges of $G$  may be left uncoloured by this process. These can be coloured arbitrarily if a full colouring is required.

Maffray and Reed~\cite{MaRe} characterised elementary graphs in a very precise way. They showed
\begin{thm}[Maffray and Reed]
G is elementary if and only if it is an augmentation of the line graph of bipartite multigraph.\qed
\end{thm}
We must describe the ``augmentation'' in this theorem. An edge $xy$ in $G$ is called \emph{flat} if $x,y$ have no common neighbour.
Then we augment the flat edge by replacing $x$ by a clique $X$, $y$ by a clique $Y$, and $xy$ by any non-empty edge set $F\subseteq (X,Y)$. That is, we replace $xy$ by a cobipartite graph called an \emph{augment}. Finally we add all edges between $X$ and $\nb(x)\sm\{y\}$ and all edges between $Y$ and $\nb(y)\sm\{x\}$, see Fig.~\ref{fig:augment}.  Then a graph~$G'$ is an augmentation of a graph~$G$ if  $G'$ can be obtained from~$G$ by applying one or more such steps to independent flat edges in $G$. See Fig.~\ref{fig:augmentation}, where the first graph is a line graph of a bipartite multigraph, as we will show below. The second and third show augmentations using the \emph{independent} flat edges $x_1y_1$ and then $x_2y_2$ and two $2\times 2$ cobipartite graphs as replacements.

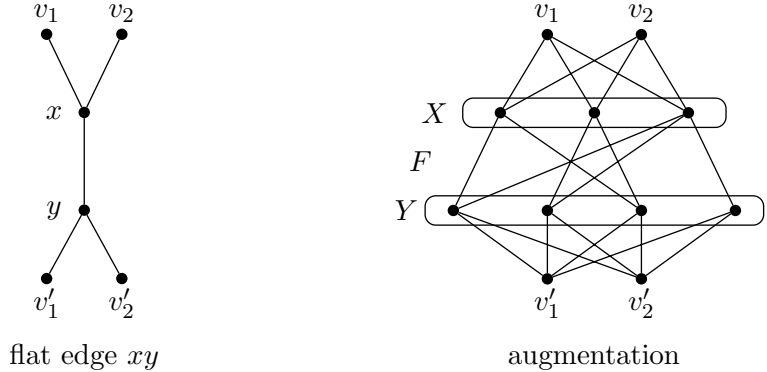
\begin{figure}
\begin{center}
\begin{tikzpicture}[line width=0.5pt,inner sep=0pt,minimum size=0mm,xscale=2,yscale=1.3,font=\small]
\draw (0.25,0.3) node[b,label=below:$\strut v'_1$] (v'1) {}(0.75,0.3) node[b,label=below:$\strut v'_2$] (v'2) {}
(0.5,2) node[b] (x) {}
(0.5,1) node[b] (y) {}
 (0.25,2.8) node[b,label=above:$\strut v_1$] (v1) {} (0.75,2.8) node[b,label=above:$\strut v_2$] (v2) {} ;
\draw (0.5,-0.5) node{flat edge $xy$};
\draw (v'1)--(y)--(v'2) (x)--(y) (v1)--(x)--(v2);
\draw (0.3,2) node {$x$} (0.3,1) node {$y$} ;
\end{tikzpicture}\hspace{3cm}
\begin{tikzpicture}[line width=0.5pt,inner sep=0pt,minimum size=0mm,xscale=2.5,yscale=1.3,font=\small]
\draw (0.25,0.3) node[b,label=below:$\strut v'_1$] (v'1) {}(0.75,0.3) node[b,label=below:$\strut v'_2$] (v'2) {}
(0,2) node[b] (x1) {} (0.5,2) node[b] (x2) {} {} (1,2) node[b] (x3) {}
(-0.25,1) node[b] (y1) {} (0.25,1) node[b] (y2) {} {} (0.75,1) node[b] (y3) {} {} (1.25,1) node[b] (y4) {}
 (0.25,2.8) node[b,label=above:$\strut v_1$] (v1) {} (0.75,2.8) node[b,label=above:$\strut v_2$] (v2) {} ;
\draw (0.5,-0.5) node{augmentation};
\draw (v'1)--(y1)--(v'2) (v'1)--(y2)--(v'2) (v'1)--(y3)--(v'2) (v'1)--(y4)--(v'2)
(y1)--(x1) (y2)--(x2) (y1)--(x3) (y2)--(x3) (x3)--(y4) (x2)--(y3)--(x1)
(v1)--(x1)--(v2) (v1)--(x2)--(v2) (v1)--(x3)--(v2) ;
\draw[rounded corners] (-0.2,1.85) rectangle (1.2,2.15) (-0.4,0.85) rectangle (1.4,1.15) ;
\draw (-0.35,2) node {$X$} (-0.5,1) node {$Y$} (-0.425,1.5) node{$F$} ;
\end{tikzpicture}
\end{center}
\caption{Augmenting a flat edge}\label{fig:augment}
\end{figure}

We observe that an augmentation of a line graph of a bipartite multigraph need not be a line graph
of a bipartite multigraph, as can be seen in Fig.~\ref{fig:augmentation} where
the two augmentations  contain a  gem which is an excluded structure for
the class of line graphs of bipartite multigraphs by the following characterisation of this class.

\begin{thm}[Maffray and Reed~\cite{MaRe}]\label{char:lgbm}
A graph is the line graph of bipartite multigraph if and only if it is (claw, gem, 4-wheel, odd hole)-free.
(See Fig.~\ref{fig:clawdiagem}.)
\end{thm}

Maffray and Reed~\cite{MaRe} show how to recover the structure of an elementary graph as the line graph of bipartite multigraph with augmented flat edges, using an elementary colouring of the graph. This can be done in $O(mn)$ time, so there is an $O(mn)$ time algorithm for determining the graph structure.
(Maffray and Reed claim only the looser bound $O(m^2)$.)

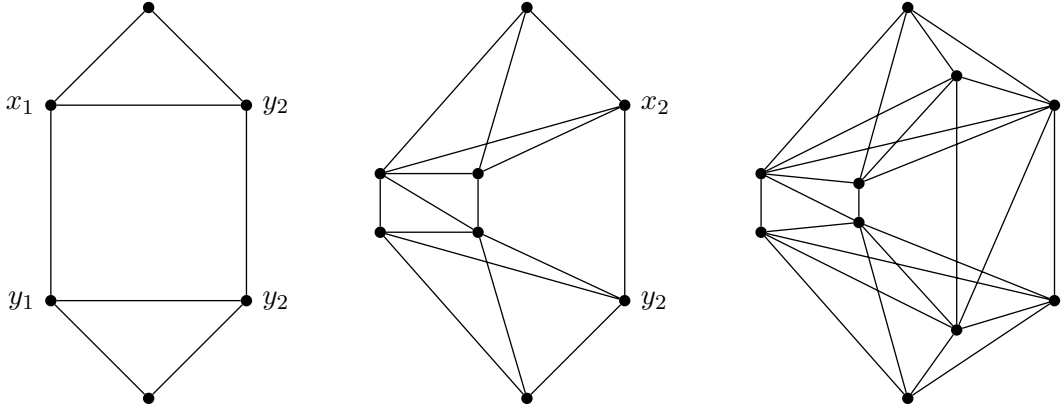
\begin{figure}[tbh]
\begin{center}
\begin{tikzpicture}[line width=0.5pt,inner sep=0pt,minimum size=0mm,scale=1.3,font=\small]
\draw (0,0) node[b] (v') {} (-1,1) node[b,label=left:$y_1\ $] (y1) {} (1,1) node[b,label=right:$\ y_2$] (y2) {}
(0,4) node[b] (v) {} (-1,3) node[b,label=left:$x_1\ $] (x1) {} (1,3) node[b,label=right:$\ y_2$] (x2) {};
\draw (v')--(y1)--(y2)--(v') (x1)--(y1) (x2)--(y2) (v)--(x1)--(x2)--(v) ;
\end{tikzpicture}\hspace{1cm}
\begin{tikzpicture}[line width=0.5pt,inner sep=0pt,minimum size=0mm,scale=1.3,font=\small]
\draw (0,0) node[b] (v') {} (-1.5,1.7) node[b] (y1) {} (-0.5,1.7) node[b] (y1') {}
(1,1) node[b,label=right:$\ y_2$] (y2) {}
(0,4) node[b] (v) {} (-1.5,2.3) node[b] (x1) {} (-0.5,2.3) node[b] (x1') {}
(1,3) node[b,label=right:$\ x_2$] (x2) {};
\draw (v')--(y1)--(y2)--(v') (x1)--(y1) (x2)--(y2) (v)--(x1)--(x2)--(v) (v)--(x1')--(x2) (v')--(y1')--(y2)
(x1)--(x1') (y1)--(y1') (x1)--(y1') (x1')--(y1');
\end{tikzpicture}\hspace{1cm}
\begin{tikzpicture}[line width=0.5pt,inner sep=0pt,minimum size=0mm,scale=1.3,font=\small]
\draw (0,0) node[b] (v') {} (-1.5,1.7) node[b] (y1) {} (-0.5,1.8) node[b] (y1') {}
(0.5,0.7) node[b] (y2) {} (1.5,1) node[b] (y2') {}
(0,4) node[b] (v) {} (-1.5,2.3) node[b] (x1) {} (-0.5,2.2) node[b] (x1') {}
(0.5,3.3) node[b] (x2) {} (1.5,3) node[b] (x2') {};
\draw (v')--(y1)--(y2)--(v') (x1)--(y1) (x2)--(y2) (v)--(x1)--(x2)--(v) (v)--(x1')--(x2) (v')--(y1')--(y2)
(x1)--(x1') (y1)--(y1') (x1)--(y1') (x1')--(y1')
(x2)--(x2') (y2)--(y2') (x2')--(y2) (x2')--(y2')
(x2')--(v) (x2')--(x1) (x2')--(x1')
(y2')--(v') (y2')--(y1) (y2')--(y1');
\end{tikzpicture}
\end{center}
\caption{Augmentation of flat edges}\label{fig:augmentation}
\end{figure}

\subsection{Counting in (claw,\,odd hole)-free graphs}\label{sec:clawfree}

From Corollary~\ref{cor:benrebea}, Lemma~\ref{lem:peculiar} and Theorem~\ref{thm:Chvatal} we  have
\begin{lemma}\label{alpha3}
  Every (claw,\,odd hole)-free graph $G$ without a clique cutset and with $\alpha(G)>3$ is elementary.
\end{lemma}
This gives us a clique cutset decomposition in which the atoms are in the hereditary class $\cA$ of graphs that are
either elementary or have $\alpha(G)\leq 3$. To apply the method
of Section~\ref{sec:cutsetcounting}, we must consider how to approximate $W(G)$ in these graphs.

\subsubsection{Computing $W(G)$ in graphs with $\alpha(G)\leq 3$}

Let $G$ be a (claw,\,odd hole)-free atom. For any $k$, we can determine $\cI_k(G)$ in $O(n^k)$ time by listing all $k$-tuples of vertices and checking which are independent in $G$. Thus we can determine $W_k(G)$ for $k=0,1,2,3,4$ in $O(n^4)$ time. If $W_4(G)>0$, we conclude that $G$ must be elementary, by Lemma~\ref{alpha3}. Otherwise, we set $W(G)=\sum_{k=0}^{3} W_k(G)$.

\subsubsection{Approximating $W(G)$ in elementary graphs}\label{elementary atoms}

If a (claw,\,odd hole)-free atom $G$ has $\alpha(G)>3$, then it is elementary by Lemma~\ref{alpha3}. We use the $O(mn)$ time algorithm of Maffray and Reed~\cite{MaRe} to identify \emph{$G$ as the line graph of a bipartite multigraph with augments}. If this algorithm fails, we conclude that the original graph was not (claw,\,odd hole)-free, and halt. Otherwise, we have an elementary atom, and we continue.

An elementary graph $G$ is not necessarily the line graph of bipartite multigraph because of the augments, as discussed in Section~\ref{sec:elementary}. However, we will replace $G$ by an equivalent $G'$, such that $G'$ is the line graph of bipartite multigraph. We do this by replacing the augments in $G$ by ``gadgets'' which are line graphs of bipartite multigraphs.

\subsubsection{Augmentation gadgets}

Suppose the augment $Z = X \cup Y$ in $G=(V,E)$, with vertex weights $w(v)$ ($v\in V$), comprises a cobipartite graph with cliques on $X, Y$ and connecting bipartite graph $(X\cup Y,F)$. Clearly $W_k(Z)\neq 0$, only for $k=0,1,2$.

If $U\subseteq V$, we will write $W_k(U)$ for $W_k(G[U])$. Then $W_1(X)=\sum_{v\in X} w(v)$, $W_1(Y)=\sum_{v\in Y} w(v)$ and $W_2(Z)=\sum_{uv\notin F} w(u)w(v)$. These are respectively the total weights of independent sets in $Z$ which involve $X$ alone, $Y$ alone, or both. We will also write $\overline{W}_2(Z)=W_1(X)W_1(Y)-W_{2}(Z)=\sum_{uv\in F} w(u)w(v)$.

Consider the gadget $Z'$ shown in Fig.~\ref{fig:gadget}, where $\rho,\sigma,\overline{\rho},\overline{\sigma}\geq 0$
are to be determined, Let $X'=\{x_1,x_2\}$, $Y'=\{y_1,y_2\}$. Note that both vertices in the clique $X'$ have the same
neighbours external to $Z'$ as all vertices in $X$ have to vertices external to $Z$, and similarly for $Y'$ and $Y$.

If we set $\overline{\rho}=W_1(X)-\rho$, then $W_1(X')=\rho+\overline{\rho}=W_1(X)$, and if we set
$\overline{\sigma}=W_1(Y)-\sigma$, then $W_1(Y')=\sigma+\overline{\sigma}=W_1(Y)$. Thus the total weight of
independent sets using $X'$ but not $Y'$ is $W_1(X)$, and the total weight of
independent sets using $Y'$, but not $X'$ is $W_1(Y)$, as required.
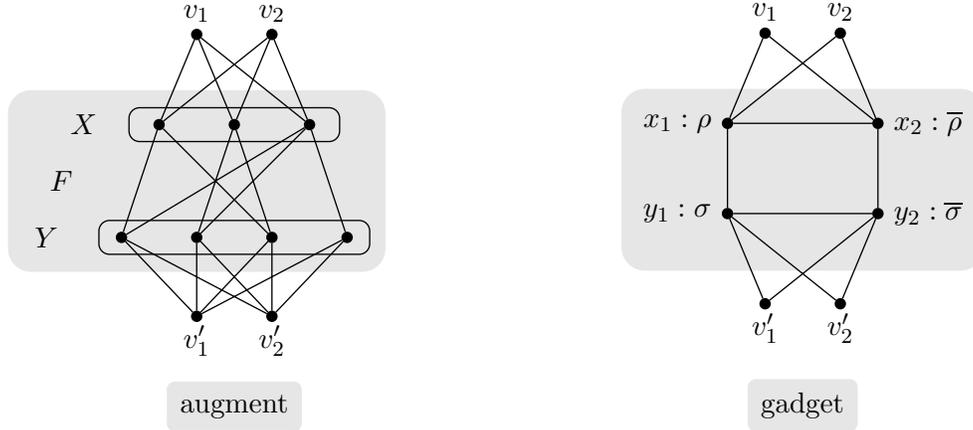
\begin{figure}[tbh]
\colorlet{lightgray}{black!10}
\begin{center}
\begin{tikzpicture}[line width=0.5pt,inner sep=0pt,minimum size=0mm,xscale=2,yscale=1.5,font=\small]
\filldraw[color=lightgray,rounded corners=3mm] (-1.0,0.7) rectangle (1.5,2.3);
\draw (0.25,0.3) node[b,label=below:$\strut v'_1$] (v'1) {}(0.75,0.3) node[b,label=below:$\strut v'_2$] (v'2) {}
(0,2) node[b] (x1) {} (0.5,2) node[b] (x2) {} {} (1,2) node[b] (x3) {}
(-0.25,1) node[b] (y1) {} (0.25,1) node[b] (y2) {} {} (0.75,1) node[b] (y3) {} {} (1.25,1) node[b] (y4) {}
 (0.25,2.8) node[b,label=above:$\strut v_1$] (v1) {} (0.75,2.8) node[b,label=above:$\strut v_2$] (v2) {} ;
\draw (0.5,-0.5) node[inner sep=5pt,fill=lightgray,rounded corners=3pt]{augment};
\draw (v'1)--(y1)--(v'2) (v'1)--(y2)--(v'2) (v'1)--(y3)--(v'2) (v'1)--(y4)--(v'2)
(y1)--(x1) (y2)--(x2) (y1)--(x3) (y2)--(x3) (x3)--(y4) (x2)--(y3)--(x1)
(v1)--(x1)--(v2) (v1)--(x2)--(v2) (v1)--(x3)--(v2) ;
\draw[rounded corners] (-0.2,1.85) rectangle (1.2,2.15) (-0.4,0.85) rectangle (1.4,1.15) ;
\draw (-0.5,2) node {$X$} (-0.75,1) node {$Y$} (-0.65,1.5) node{$F$} ;
\end{tikzpicture}\hspace{3cm}
\begin{tikzpicture}[line width=0.5pt,inner sep=0pt,minimum size=0mm,xscale=2,yscale=1.5,font=\small]
\filldraw[color=lightgray,rounded corners=3mm] (-0.7,0.7) rectangle (1.7,2.3);
\draw (0.25,0.4) node[b,label=below:$\strut v'_1$] (v'1) {}(0.75,0.4) node[b,label=below:$\strut v'_2$] (v'2) {}
(0,2) node[b,label=left:$x_1: \rho\ $] (x1) {} (1,2) node[b,label=right:$\ x_2: \overline{\rho}$] (x2) {}
 (0,1.2) node[b,label=left:$y_1: \sigma\ $] (y1) {} (1,1.2) node[b,label=right:$\ y_2: \overline{\sigma}$] (y2) {}
 (0.25,2.8) node[b,label=above:$\strut v_1$] (v1) {} (0.75,2.8) node[b,label=above:$\strut v_2$] (v2) {} ;
\draw (0.5,-0.5) node[inner sep=5pt,fill=lightgray,rounded corners=3pt]{gadget}; 
\draw (v'1)--(y1)--(v'2) (v'1)--(y2)--(v'2)
(x1)--(x2) (y1)--(x1) (y1)--(y2)  (y2)--(x2)
(v1)--(x1)--(v2) (v1)--(x2)--(v2)   ;
\end{tikzpicture}
\end{center}
\caption{Augmentation and equivalent gadget with vertex weights}\label{fig:gadget}
\end{figure}

The weight of independent sets using both $X'$ and $Y'$ is
\[ W_2(Z')=\rho\overline{\sigma}+\overline{\rho}\sigma= \rho (W_1(Y)-\sigma)+\sigma(W_1(X)-\rho)= W_2(Z)\;, \]
again as required, provided
\[ \rho\ =\ \frac{W_2(Z)-\sigma W_1(X)}{W_1(Y)-2\sigma}\;. \]
It is convenient to break symmetry by requiring $\sigma<\overline{\sigma}$, making the denominator of the above fraction positive.
We must also have $\rho,\overline{\rho}\geq 0$, so $0\leq \rho \leq W_1(X)$. Thus we require
\[0\leq\sigma \leq \min\{W_2(Z),\overline{W}_2(Z)\}/W_1(X).\] Otherwise we can choose $\sigma$ arbitrarily. Then the gadget $Z'$ is equivalent to $Z$.
The particular choice $\sigma=0$ deletes $y_1$ and its incident edges, and gives an even smaller gadget. Note that $Z'$ is itself an augment,
equivalent to $Z$ for computing $W_k(G)$ for any $0\leq k\leq\alpha(G)$.

The gadget $Z'$ is the line graph of a simple bipartite graph, see Fig.~\ref{fig:equivalentLG}.
\begin{figure}[tbh]
\begin{center}
\begin{tikzpicture}[line width=0.5pt,inner sep=0pt,minimum size=0mm,xscale=1.5,yscale=1.25,font=\small]
\draw (0.25,0) node[b] (v'1) {} (0.75,0) node[b] (v'2) {}
(0.5,2.1) node[b] (x1) {} (0,1.4) node[b] (x2) {}
 (1,1.4) node[b] (y1) {} (0.5,0.7) node[b] (y2) {}
 (0.25,2.8) node[b] (v1) {}  (0.75,2.8) node[b] (v2) {} ;
\draw (v1)--node[left]{$v_1\ $}(x1) (v2)--node[right]{$\ v_2$}(x1)
(x1)--node[left]{$x_1\ $}(x2)
(y1)--node[right]{$\ x_2$}(x1)  (y1)--node[right]{$\ y_2$}(y2)
(y2)--node[left]{$y_1\ $}(x2) (y2)--node[right]{$\ v'_2\ $}(v'2)
(y2)--node[left]{$v'_1\ $}(v'1)   ;
\end{tikzpicture}
\end{center}
\caption{Root graph of the gadget depicted in Fig.~\ref{fig:gadget}}\label{fig:equivalentLG}
\end{figure}
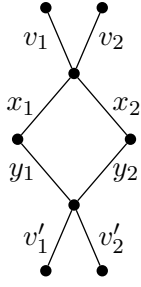
Every augment $Z$ that has less vertices than $Z'$ is also line graph of
a suitable bipartite graph.
So using this gadget to replace all augments on at least four vertices in $G$ will result in it becoming the line graph $G'$ of bipartite multigraph, as required. Moreover, the size of $G'$ does not exceed the size of $G$, which becomes relevant in the following analysis.

\subsubsection{Reduction to permanent approximation}\label{sec:permanent}

We now need to determine $W_k(G')$, where $G'$ is the line graph of a bipartite multigraph. We determine its root multigraph $G''=L^{-1}(G')$ and verify that it is bipartite
in $O(m)$ time, using (for example) the algorithm of Lehot~\cite{Lehot}. The vertex weights in the line graph 
$G'$ become edge weights in $G''$, and independent sets become matchings. The graph  $G''$ may have parallel edges,
but for our purposes we can reduce this to a simple edge-weighted bipartite graph $G^*$ by adding the weights on parallel edges.
Now $M_k(G^*)=W_k(G')$ will be the total weight of all \emph{matchings} of size $k$ in $G^*$.

Suppose $G^*=(V_1\cup V_2,E^*)$, where $n_1=|V_1|$, $n_2=|V_2|$. We wish to use the permanent algorithm of~\cite{JSV} to determine
$M_k(G^*)$. However, this algorithm only computes $M_n(G^*)$ for the perfect matching case $n_1=n_2=n$. It is possible the algorithm of~\cite{JSV} can be modified to the general case, but the general case can be reduced to the permanent, as follows.

To determine $M_k(G^*)$, let $n'_1=n_2-k$, $n'_2=n_1-k$. We add a set $V'_1$ of $n'_1$ vertices to $G^*$, and the edges of a complete bipartite graph $K_{n'_1,n_2}=(V'_1\cup V_2,V'_1\times V_2)$, and add a set $V'_2$ of $n'_2$ vertices and the edges of a complete bipartite graph $K_{n_1,n'_2}=(V_1\cup V'_2,V_1\times V'_2)$. See Fig.~\ref{fig:equivalentBG}, where $n_1=5$, $n_2=4$, $k=2$. The weights assigned to the added edges are all $1$. Let this weighted graph be $G^+=(V_1^+\cup  V_2^+,E^+)$, where $V_1^+=V_1\cup V'_1$, $V_2^+=V_2\cup V'_2$, and $E^+$ is $E^*$ plus the edges of the complete bipartite graphs. Let $n^+=|V_1^+|=|V_2^+|=n_1+n_2-k=n-k$.

Now observe that there is a correspondence between matchings of size $k$ in $G^*$ and perfect matchings in $G^+$. For each $k$-matching in $G^*$, we match the $n_1-k$ unmatched vertices in $V_1$ with vertices in $V'_2$, and we match the $n_2-k$ unmatched vertices in $V_2$ with vertices in $V'_1$. Given a perfect matching $M^+$ in $G^+$, we can uniquely recover a $k$-matching $M^*$ in $G^*$ of the same weight. However, there are $n'_1!\hspace{1pt}n'_2!=(n_1-k)!\hspace{1pt}(n_2-k)!$ matchings $M^+$ corresponding to any $M^*$. Thus $M_{n^+}(G^+)=(n_1-k)!\hspace{1pt}(n_2-k)!\hspace{1pt}M_k(G^*)$.

Thus our algorithm will use the permanent method of~\cite{JSV} to compute $M_{n^+}(G^+)$, and then divide this by $(n_1-k)!\hspace{1pt}(n_2-k)!$ to obtain $M_k(G^*)$. Thus we can determine $W_k(G')$ for any elementary graph $G'$ and any $0\leq k\leq\alpha(G')$.
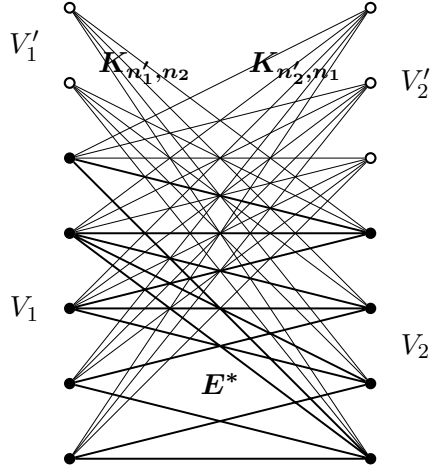
\begin{figure}
\begin{center}
\begin{tikzpicture}[line width=0.25pt,inner sep=0pt,minimum size=0mm,xscale=2,yscale=2,font=\small]
\foreach \x in {1,...,5}{\node[b] (0\x) at (0,0.5*\x){} ;};
\foreach \x in {6,...,7}{\node[w] (0\x) at (0,0.5*\x){} ;};
\foreach \x in {1,...,4}{\node[b] (1\x) at (2,0.5*\x){} ;};
\foreach \x in {5,...,7}{\node[w] (1\x) at (2,0.5*\x){} ;};
\foreach \x in {1,...,5}\foreach \y in {5,6,7}{\draw[line width=0.25pt] (0\x)--(1\y);} ;
\foreach \x in {1,...,4}\foreach \y in {6,7}{\draw[line width=0.25pt] (1\x)--(0\y);} ;
\draw[line width=0.75pt] (12)--(01)--(11)--(02) (04)--(13)--(02) (13)--(03)--(12)--(04)
(05)--(14)--(03) (14)--(04)--(11)--(05) ;
\node at (-0.3,1.5) {$V_1$}; \node at (2.3,1.25) {$V_2$};
\node at (-0.3,3.25) {$V'_1$}; \node at (2.3,3) {$V'_2$};
\node at (1,1) {$\boldsymbol{E^*}$} ; \node at (0.5,3.1) {$\boldsymbol{K_{n'_1,n_2}}$} ; \node at (1.5,3.1) {$\boldsymbol{K_{n'_2,n_1}}$} ;
\end{tikzpicture}
\end{center}
\caption{Equivalent permanent problem}\label{fig:equivalentBG}
\end{figure}
We then compute $W(G')=\sum_{k=1}^{\alpha(G')} W_k(G')$.
This has time complexity is $O(nT_{\mathrm{P}}(n,\epsilon))$, where $T_{\mathrm{P}}(n,\epsilon)$ is the time
to approximate the total weight of perfect matchings in an $n\times n$ bipartite graph with relative error $1\pm\epsilon$.  The clique decomposition in Section~\ref{sec:cutsetcounting}
gives another $n$ factor, so the overall
time complexity of approximating $W(G)$ in a (claw,\,odd hole)-free graph $G$ with $n$ vertices
is $T_{\cC}(n,\epsilon)=O(n^2T_{\mathrm{P}}(n,\epsilon n^{-2}))$, where the accuracy parameter comes from the error analysis in Section~\ref{sec:clawerror}.
The best  bound for $T_{\mathrm{P}}(n,\epsilon)$ known is $O(n^7\log^4 n+n^6\log^5(n)\epsilon^{-2})$~\cite{BSVV}.
Thus the overall time complexity for (claw,\,odd hole)-free graphs is
$T_{\cC}(n,\epsilon)=O(n^{12}\log^5(n)\epsilon^{-2})$.  This analysis is clearly loose and could be tightened.  However, without a radical improvement in the bound for $T_{\mathrm{P}}(n,\epsilon)$, the overall time complexity cannot be improved to anything practically relevant.

\subsection{Approximating $W_k(G)$}\label{sec:claw W_k}

The freedom to use very large vertex weights allows us to approximate $W_k(G)$ for any $0\leq k\leq\alpha(G)$.
For $k<\alpha(G)$, we would need to use the algorithm described in~\cite{DyGrMu}, which is an improvement of the approach
of that in~\cite{JS}. This method is based on the fact that the \emph{independence polynomial}
\begin{equation}
P_G(\lambda)=\sum_{i=0}^{\alpha(G)} \lambda^k\,W_k(G)\label{eq:indpoly}
\end{equation}
has only real negative roots when $G$ is claw-free. This was proved in~\cite{CS} for unit weights, and extended to general weights
in~\cite{DyGrMu}. However, the algorithm requires the reduction from approximate counting to approximate
random generation~\cite{JVV}. Therefore, we will not give further details here, but it is a
straightforward application of the above algorithm for $W(G)$.

However, for $\alpha=\alpha(G)$ the total weight $W_\alpha(G)$ of maximum independent sets, which correspond exactly to
perfect matchings in a graph when $G$ is a line graph, there is a simpler approach, which we will describe.
Note that this gives a complete generalisation of the result of~\cite{JSV}.

We have a multiplier $\lambda$ for every vertex weight, as in \eqref{eq:indpoly} above,
so $w(v)\gets \lambda w(v)$ for all $v\in V$. Then $W(G)$ becomes\\[-10pt]
{\addtolength{\jot}{-5pt}
\begin{align*}
W(\lambda G)&=\sum_{I\in\cI(G)}w(I)\,\lambda^{|I|}=\sum_{k=0}^{\alpha(G)} W_k(G)\,\lambda^k=P_G(\lambda)\;,
\intertext{Thus, if $\lambda \geq 1$ and $\alpha=\alpha(G)$,}
W_\alpha(G)\,\lambda^\alpha & \leq W(\lambda G) < W(G)\lambda^{\alpha-1} + W_\alpha(G)\,\lambda^\alpha
\intertext{so}
W_\alpha(G) & \leq W(\lambda G)/\lambda^\alpha < W_\alpha(G) + W(G)/\lambda\;.
\end{align*}}%
So we need $\lambda \geq W(G)/\varepsilon W_\alpha(G)$ to achieve relative error $\epsilon$. Let $w_{\min}\leq w(v)\leq w_{\max}$ for all $v\in V$. Then $W(G) \leq 2^n w_{\max}^\alpha$ and $W_\alpha(G)\geq w_{\min}^\alpha$, since $G$ has at most $2^n$ independent sets and at least one of size $\alpha$. Thus it suffices to take
$\lambda \geq 2^n (w_{\max}/w_{\min})^\alpha/\epsilon$ in order that $W(\lambda G)/\lambda^\alpha$ approximates $W_\alpha(G)$ with relative error $\varepsilon$. Since the time complexity of the algorithm is polynomial in $\log \lambda$, it is clearly polynomial in $n$ and the bit size of the $w(v)$'s, as required.

\section{Approximating $W(G)$ in (fork,\,odd hole)-free graphs}\label{sec:forkcount}

We will extend the result for claw-free graphs to fork-free graphs using modular decomposition, as described in
Section~\ref{sec:moddecomp}. Our algorithm is inspired by Lozin and Milani\v{c}'s~\cite{LozMil} approach to computing the
maximum weight independent set problem using modular decomposition. Again, we will first describe the modular
decomposition approach in a general context.

\subsection{Modular decomposition}\label{sec:moddecomp}

\emph{Modules} were first introduced by Gallai~\cite[Thm.\,3.1.2]{gallai,MaPr}, using different terminology.
If $S\subseteq V$ in $G=(V,E)$, we say that any vertex $x\in V\sm S$ \emph{distinguishes} $S$ if there exist
$u,v\in S$ with $ux\in E$, $vx\notin E$. Then a set $M \subseteq V$ is a \emph{module} of $G$ if no vertex of
$V\sm M$ distinguishes it. Alternatively, $M$ is a module if $\nb(u) \sm M = \nb(v) \sm M$ holds for all $u,v \in M$.
Thus $\es$, $V$ and all the singletons $\{v\}$ ($v \in V$) are modules of $G$. These are
the \emph{trivial} modules; all other modules are \emph{nontrivial}.

Another way of defining a module $M$ is that every vertex $v\in V\sm M$ must be either complete or anticomplete to $M$.
It follows that $M$ is a module in $G$ if and only if it is a module in $\overline{G}$, since $v$ is complete to $M$ in $G$
if and only if it is anticomplete to $M$ in $\overline{G}$. It is also easy to show that the modules of $G$ are closed
under intersection. However, if $M_1,M_2$ are modules, then $M_1\cup M_2$ is only guaranteed to be a module if
$M_1\cap M_2\neq\es$. Otherwise some $v$ could be complete to $M_1$ and anticomplete to $M_2$,
so neither complete nor anticomplete to $M_1\cup M_2$. Modules are also not generally closed under complementation,
since if $u\notin M$ is complete to $M$ and $v\notin M$ is anticomplete to $M$, then any vertex of $M$ distinguishes $V\sm M$.

\begin{obs}\label{obs:subgraph}
  If $M$ is a module of $G=(V,E)$ and $U \subseteq V$ then $M \cap U$ is a module of $G[U]$.
\end{obs}
\begin{proof}
  Otherwise, two vertices $u,v \in M \cap U$ distinguishable by $x \in U \sm M$
  belong to $M$, that is $u,v \in M$, and are distinguished by $x \in V \sm M$.
\end{proof}

If $M\ne\es$ is a module of $G$ then $G/M$ denotes the graph obtained from $G$ by
\emph{contracting} $M$ to a single vertex, with the same adjacencies in $G/M$ as all vertices in $M$.
We will label this vertex as $v_M$ in $G/M$. For $|M|\leq 1$ let $G/M=G$. Note that
\begin{obs}\label{obs:contraction}
  If $M$ is a module in $G\in\cC$, for some hereditary class $\cC$, then $G/M\in\cC$.
\end{obs}
\begin{proof}
  Contracting $M$ deletes all but one of its vertices, and relabels the remaining vertex $v_M$.
  Since $\cC$ is a class of unlabelled graphs, $G/M\in\cC$ by heredity.
\end{proof}
If $M$ is a module of $G$, and $M'$ is a module of $G[M]$, then $M'$ is also a module of $G$. 

A module that does not overlap with any other module is strong \cite{Kelly},
more formally, $M \subseteq V$ is a \emph{strong module} of $G=(V,E)$ if $M$
is a module of $G$ and for all modules $M'$ of $G$ we have $M \subseteq M'$ or
$M \supseteq M'$ or $M \cap M' = \es$. Every trivial module of $G$ is
strong. $G$ is a \emph{prime graph} if every strong module of $G$ is trivial.
Like connectedness, primeness is an intrinsic property of the graph, but not 
a hereditary property. 
For example, none of the graphs in Fig.~\ref{fig:clawdiagem} is prime.

The strong modules of $G=(V,E)$ are partially ordered by set inclusion. The
unique top element of this poset is $V$, and $\es$ is the unique bottom element.
The layer above $\es$ consists of the singletons $\{v\}$ for all $v \in V$.
The strong modules in the next layer up are called \emph{prime modules} of $G$.
A prime module $M$ of $G$ induces a prime subgraph $G[M]$.

Let $M_1$ and $M_2$ be strong modules such that $M_1 \subset M_2$. We say
$M_2$ \emph{covers} $M_1$ if, for all strong modules $M$, $M_1 \subseteq M$
and $M \subseteq M_2$ imply $M_1=M$ or $M=M_2$. If $|V|>1$ then the strong
modules of $G=(V,E)$ except those of size at most one are the nodes of a tree
rooted at $V$, where the arcs of the tree are given by the cover relation.
Note that the singleton modules are often included as leaves in this tree, but
we will not do so here. See Fig.~\ref{fig:moduledecomposition}, where
$M_7=V$. We will call this the \emph{standard decomposition tree}. 
Equivalent definitions exist, see~\cite{HPsurvey}.
Several algorithms are known to compute the standard decomposition tree
in linear time, see for example~\cite{TCHP}.

We will use this tree in an equivalent form.
Let us number the modules $M_1,M_2,\ldots,M_h=V$, according to a \emph{postorder} on the standard tree.
This order places all the descendants of a vertex before the vertex itself, as in Fig.~\ref{fig:moduledecomposition}. (See, for example,~\cite[Ch.\,3]{Valiente}.)
We will call this the \emph{extended} decomposition tree.
Here $G_0=G$, and $G_i=G_{i-1}/\tilde{M}_i$ $(i\in[h])$, where $\tilde{M}_i$ is
$M_i$ after $M_1,\ldots,M_{i-1}$ have been contracted in order to single vertices.
We will denote this by $G_{i-1}=G/(M_1,M_2,\ldots,M_{i-1})$, and similarly
and hence $\tilde{M}_i=M_i/(M_1,M_2,\ldots,M_{i-1})$ is a module in $G_{i-1}$.
We can represent the extended decomposition tree as shown in
Fig.~\ref{fig:moduledecomposition}, where $G_{i-1}$ ($i\in[h]$) are the internal vertices, and
$\tilde{M}_1,\tilde{M}_2,\ldots,\tilde{M}_h$ and $G_h$ are the leaves.

Note that $\tilde{M}_i$ is a prime module in $G_{i-1}$, and in particular $G_{i-1}[\tilde{M}_i]$ is a prime graph.
Also $G_h$ is a single vertex, since $M_h=V$ has been contracted to a single vertex.
Note also that $\tilde{M}_i$ is isomorphic to the graph obtained by contracting only its children
in the standard tree, but its vertex weights require contracting the whole subtree of which it is the root.

Let $t_i=|\tilde{M}_i|$, so $2\leq t_i< n$. Then $|V(G_0)|=n$, and $|V(G_i)|=|V(G_{i-1})|-t_i+1$,
so $1=|V(G_h)|=n-\sum_{i=1}^h t_i+h\leq n-2h+h=n-h$. Thus $h\leq n-1$, so
the decomposition tree contains at most $(n-1)$ modules. Also $n-\sum_{i=1}^h t_i+h=1$ implies
$\sum_{i=1}^h t_i =n-h-1\leq n-2$. Thus the extended decomposition tree can be
represented explicitly in $O(n)$ space.

\begin{figure}
\begin{center}
\begin{tikzpicture}[line width=0.75pt,minimum size=0mm,xscale=1.3,yscale=1.5,font=\small]
\draw (0,0) node[circle,draw,inner sep=1pt] (M1) {$M_1$} (1,0) node[circle,draw,inner sep=1pt] (M2) {$M_2$}
(3.5,0) node[circle,draw,inner sep=1pt] (M5) {$M_5$} ;
\draw (0.5,1) node[circle,draw,inner sep=1pt] (M3) {$M_3$} (2,1) node[circle,draw,inner sep=1pt] (M4) {$M_4$}
(3.5,1) node[circle,draw,inner sep=1pt] (M6) {$M_6$} (2,2) node[circle,draw,inner sep=1pt] (M7) {$M_7$} ;
\draw (M7)--(M3)--(M2) (M7)--(M4) (M7)--(M6) (M3)--(M1) (M6)--(M5) ;
\end{tikzpicture}
\hspace*{2cm}
\begin{tikzpicture}[line width=0.75pt,minimum size=0mm,xscale=1,yscale=2.5,font=\small]
\draw
(0,1.7) node[circle,draw,inner sep=1pt] (G0) {$G_0$}+(0,-0.5) node[circle,draw,inner sep=0pt] (M1) {$\tilde{M}_1$}
(1,1.6) node[circle,draw,inner sep=1pt] (G1) {$G_1$}+(0,-0.5) node[circle,draw,inner sep=0pt] (M2) {$\tilde{M}_2$}
(2,1.5) node[circle,draw,inner sep=1pt] (G2) {$G_2$}+(0,-0.5) node[circle,draw,inner sep=0pt] (M3) {$\tilde{M}_3$}
(3,1.4) node[circle,draw,inner sep=1pt] (G3) {$G_3$}+(0,-0.5) node[circle,draw,inner sep=0pt] (M4) {$\tilde{M}_4$}
(4,1.3) node[circle,draw,inner sep=1pt] (G4) {$G_4$}+(0,-0.5) node[circle,draw,inner sep=0pt] (M5) {$\tilde{M}_5$}
(5,1.2) node[circle,draw,inner sep=1pt] (G5) {$G_5$}+(0,-0.5) node[circle,draw,inner sep=0pt] (M6) {$\tilde{M}_6$}
(6,1.1) node[circle,draw,inner sep=1pt] (G6) {$G_6$}+(0,-0.5) node[circle,draw,inner sep=0pt] (M7) {$\tilde{M}_7$}
(7,1) node[circle,draw,inner sep=1pt] (G7) {$G_7$};
\draw (G0)--(G1)--(G2)--(G3)--(G4)--(G5)--(G6)--(G7);
\draw (G0)--(M1) (G1)--(M2) (G2)--(M3) (G3)--(M4) (G4)--(M5) (G5)--(M6) (G6)--(M7);
\end{tikzpicture}
\end{center}
\caption{ A standard modular decomposition tree and its extended tree}\label{fig:moduledecomposition}
\end{figure}
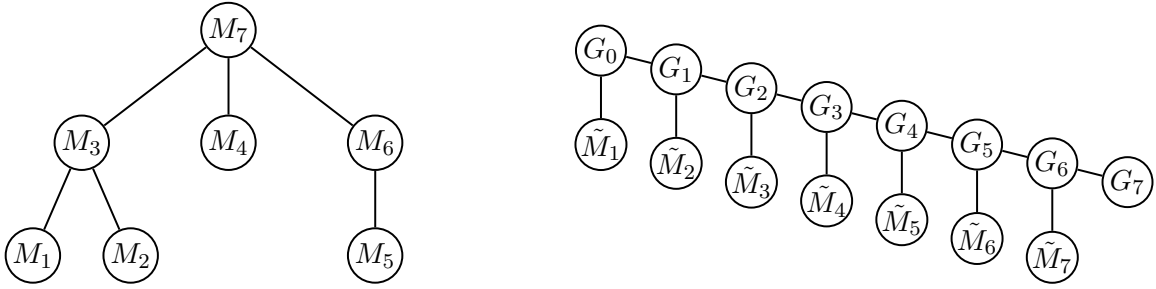

We can compute the extended tree from the standard tree in a further $O(n)$ time, using postorder tree traversal~\cite{Valiente}. We can contract modules and form the modules $\tilde{M}_i$, during the traversal. Since $\sum_{i=1}|\tilde{M_i}|<n$, the additional time complexity for contracting modules is also $O(n)$, so the total remains $O(m)$. Of course, this excludes the time to compute the vertex weight of $v_{M_i}$ in $G_i$, as will be detailed in Section~\ref{sec:modcount} below. However, these computations can also be integrated into the tree traversal. Performing the algorithm this way, the extended tree is purely a useful
notional device, and never computed explicitly.

Of course, we could compute the extended tree explicitly by successively finding a prime module and contracting it, until the contracted graph is prime. While this may be conceptually simpler, it is  computationally inefficient. Finding any prime module appears to be an $\Omega(m)$ computation, so the time complexity of producing the whole extended tree becomes $\Omega(mn)$. The inefficiency clearly results from discarding information gained in earlier searches when carrying out the later searches.

We will make use of the following.
\begin{obs}\label{obs:isomorphic}
  Each of the modules $\tilde{M}_i$ ($i\in[h]$) in the extended tree is isomorphic to a prime subgraph of $G$.
\end{obs}
\begin{proof} Since $\tilde{M}_i$ is $M_i$ with all its submodules contracted, it follows from Observation~\ref{obs:contraction}, that $\tilde{M}_i$ is isomorphic to some subgraph $\tilde{M}'_i$ of $G$. Note that $\tilde{M}'_i$ is not unique. As observed above, each of the $\tilde{M}_i$ is prime, so $\tilde{M}'_i$ is also prime. Since they are vertex-weighted graphs, the isomorphism  between $\tilde{M}_i$ and $\tilde{M}'_i$ is in the sense defined in Section~\ref{sec:prelim}.
\end{proof}

\subsection{Approximating $W(G)$ using modular decomposition}\label{sec:modcount}

Let $\cC$ be a hereditary class, and $\cP\subseteq \cC$, the (non-hereditary) class of prime graphs in $\cC$. Let $T_{\cC}(n)$ be the time to compute $W(G)$ for any connected $n$-vertex graph $G\in\cC$, and let $T_{\cP}(n)$ bound the time to compute $W(G)$ for any $n$-vertex prime graph $G\in\cC$. We may assume that $T_{\cP}(n)$ is a monotonically increasing function that is linear or convex. We will show $T_{\cC}(n)\leq T_{\cP}(n)+O(m)$. This strengthens the result of~\cite[Thm.\,1]{LozMil}, with an easier proof.

We use the notation of Section~\ref{sec:moddecomp}. We construct the extended decomposition tree as described in Section~\ref{sec:moddecomp}, and
begin with $G_0=G$. At step $i$, we contract the module $\tilde{M_i}$
in $G_{i-1}$ to give $G_i$, 
giving $v_{\tilde{M_i}}$, the vertex that represents $\tilde{M_i}$ in $G_i$, weight $W(\tilde{M_i})$. Then
$W(G_i)=W(G_{i-1})$,
since the set $\tilde{M_i}$ has the same neighbourhood
in $G_{i-1}$ as the vertex $\tilde{M_i}$ in $G_i$.
Thus, by induction, $W(G)=W(G_0)=W(G_h)=w(v)$ for the unique vertex
$v \in V(G_h)$.

If $h=1$, $T_{\cC}(n)= T_{\cP}(n)$. Otherwise, $2\leq h\leq n-1$ so, omitting the time to compute the modular decomposition,
we have
\begin{align*}
  T_{\cC}(n)\,\leq &\, \max \big\{ \textstyle\sum_{i=1}^h T_{\cP}(t_i) :
  \textstyle\sum_{i=1}^h t_i=n+h, 2\leq t_i\leq n, i\in [h]\big\} \\
    \leq &\, T_{\cP}(n-h+2)+(h-1)T_{\cP}(2),\\
    \leq &\, T_{\cP}(n)+(n-2)T_{\cP}(2)\,=\, T_{\cP}(n)+O(n),
\end{align*}
where the second line follows from the first since $T_{\cP}$ is convex, so $\sum_{i=1}^h T_{\cP}(t_i)$
is maximised by setting $t_1=n-h+2$, $t_i=2$, $i=2,\ldots,h$. The third line follows from the second
because $T_{\cP}$ is increasing, so $T_{\cP}(n-h+2)\leq T_{\cP}(n)$ for $h\geq 2$.
Adding the $O(m)$ time to compute the modular decomposition~\cite{TCHP}, we have $T_{\cC}(n) \leq T_{\cP}(n)+O(m)$.
Thus we can approximate $W(G)$ in any graph in $\cC$, with only an $O(m)$ overhead, if we can approximate it in all the prime graphs in $\cC$.

Note that this analysis deals only with applications of the algorithm for $\cP$, as does that in~\cite[Thm.\,1]{LozMil}. It ignores the effect of the bit-size of the vertex weights on $T_{\cP}(n)$. This distinction is not so important for optimisation, but is much more important for counting, since contracting modules can cause exponential growth in the weights. However, we will not pursue this issue further.

\subsection{Structure of fork-free graphs}
Lozin and Milani\v{c}~\cite{LozMil} used a modular decomposition approach to determine the maximum weight independent
set in a fork-free graph. However, there seems to be a flaw in their algorithm and its analysis.
Consequently, we will re-work most of their development, in addition to extending it from optimisation to counting.

The approach of~\cite{LozMil} is based on a structural result given in ~\cite[Thm.\,3]{LozMil}.
We begin with a more useful version of this theorem, the original being too weak for its application.
We first repeat two structural lemmas from~\cite{LozMil}. We will also make use of the following simple
observation, which was used as the basis of the algorithm in~\cite{TCHP}.
\begin{obs}
If $v$ is any vertex of $G$, and $M$ is a module not containing $v$,
then either $M\subseteq \nb(v)$ or $M\subseteq V\sm \nb[v]$.\label{obs3}
\end{obs}
\begin{proof}
  Otherwise, $v$ distinguishes $M$, contradicting it being a module.
\end{proof}
\begin{lemma}[see~\cite{LozMil}, Thm.\,3 and also \cite{BHV04}]\label{lem:LM0}
If a prime fork-free graph contains a claw, then it contains one of the graphs $H_1,\ldots,H_5$
(see Fig.~\ref{fig:H1-5}).\qed
\end{lemma}

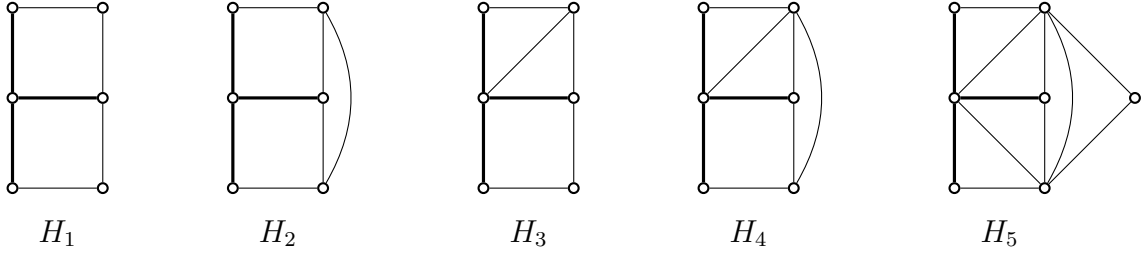
\begin{figure}[htbp]
  \newcommand{\LMH}[2]{
    \hspace*{\fill}
    \begin{tikzpicture}[scale=0.6]
      \node[w] (a) at (0,5) {}; \node[w] (b) at (2,5) {};
      \node[w] (c) at (0,3) {}; \node[w] (d) at (2,3) {};
      \node[w] (e) at (0,1) {}; \node[w] (f) at (2,1) {};
      \node    (H) at (1,0) {$H_{#1}$};
      \draw[black, very thick] (a)--(c)--(e)  (c)--(d);
      \draw (a)--(b)--(d)--(f)--(e);
      #2
    \end{tikzpicture}
  \hspace*{\fill}
  }
  \LMH{1}{}
  \LMH{2}{\draw (b) to[bend left] (f);}
  \LMH{3}{\draw (c)--(b);}
  \LMH{4}{\draw (b) to[bend left] (f); \draw (c)--(b);}
  \LMH{5}{\node[w] (g) at (4,3) {};
          \draw (b) to[bend left] (f); \draw (c)--(b)--(g)--(f)--(c);}
  \caption{The five minimal fork-free prime graphs extending a claw.}
  \label{fig:H1-5}
\end{figure}

\begin{lemma}[\cite{LozMil}, Lemma 1]\label{lem:LM1}
Let $G$ be a fork-free graph and let $v$ be any vertex of $G$. Assume that
$H\in \{ H_1, \ldots ,H_5\}$ is an induced subgraph of $\Gv$.
Then no neighbour of $v$ distinguishes $V(H)$.\qed
\end{lemma}

Note that we could omit $H_2$ for our application, because it contains a 5-hole, but we give the result for general claw-free graphs, as used in~\cite{LozMil}, since what follows could be used for computing a maximum weight independent set in a claw-free graph.

The statement and proof of the following theorem modifies the weaker result of ~\cite[Thm.\,3]{LozMil}.

\begin{thm}\label{thm:prime}
Let $G$ be a fork-free graph and $v$ a vertex of $G$.
If $G$ is prime and $M$ is a prime subgraph of $\Gv$, then $M$ is claw-free.
\end{thm}

\begin{proof}
Assume by contradiction that $M$ contains a claw.
Then by Lemma \ref{lem:LM0}, $M$ contains $H$, one of the graphs $H_1, \ldots ,H_{5}$.
Hence, by Lemma \ref{lem:LM1}, $\nb(v)$ can be partitioned into sets $Y$ and $Z$, such that
$Y$ is anticomplete to $H$ and $Z$ is complete to $H$.

Let $W$ be an (inclusionwise) maximal subset of vertices of $\Gv$ satisfying the following properties:
\begin{itemize}[topsep=0pt,itemsep=0pt]
\item[(i)] $V(H)\subseteq W$,
\item[(ii)] $G[W]$ is connected,
\item[(iii)] $\overline{G[W]}$ is connected,
\item[(iv)] $Z$ is complete to $W$,
\item[(v)] $Y$ is anticomplete to $W$.
\end{itemize}

Note that such a set $W$ exists since $V(H)$ satisfies all these properties.
Clearly, $5<|V(H)|\leq |W|< |V(G)|$.
Since $G$ is prime, $W$ cannot be a nontrivial module of $G$, and hence some $u\in V(G)\sm W$
distinguishes $W$. Note that $u\not=v$. We will obtain a contradiction (to the existence of $W$) by showing that the set $W'=W\cup \{ u\}$ also satisfies
properties (i)--(v).

Since $W$ satisfies (iv) and (v), $u\in V\sm \nb(v)$ and hence $W'\subseteq V\sm \nb(v)$.
Clearly $W'$ satisfies~(i).
Since $W$ satisfies (ii) and (iii), and since $u$ has both a neighbour and a non-neighbour in $W$,
it follows that $W'$ satisfies (ii) and (iii).

Suppose that $u$ has a non-neighbour $z\in Z$.
Since $u$ distinguishes $W$ and $\overline{G[W]}$ is connected, $u$ distinguishes a pair of nonadjacent vertices
$w_1,w_2\in W$. But then $\{ u,w_1,z,w_2,v\}$ induces a fork, a contradiction. Therefore, $Z$ is complete to $W'$ and hence $W'$ satisfies (iv).

Finally, suppose that $u$ has a neighbour $y\in Y$. Since $G[W]$ is connected, and $u$ distinguishes $W$, there is a shortest path $P=(v_0,\dots,v_k)$ connecting $V(H)$ and $u$ in $G[W']$ with
$v_0\in V(H)$ and $v_k=u$. Let $v_{k+1}=y$ and $v_{k+2}=v$.
Note that $v$ is anticomplete to $V(P)$, and $y$ is anticomplete to $V(P)\sm\{v_k\}$,
and hence $(v_0,\dots,v_{k+2})$ is a chordless path.
Since $v_2$ has no neighbour in $H$, by Lemma \ref{lem:LM1}, $v_1$ is complete to $V(H)$.
But then any two nonadjacent vertices of $H$, together with $v_1,v_2,v_3$ induce a fork in $G$,
 a contradiction. Therefore, $Y$ is complete to $W'$ and hence $W'$ satisfies (v).
\end{proof}

Note that Theorem~\ref{thm:prime} cannot obviously be strengthened. If $G$ is a fork-free prime graph and $v$ a vertex of $G$, then $\Gv$ is not necessarily claw-free nor prime.
Consider, for example, a $n\times n$ complete bipartite graph
with a perfect matching removed. This is easily shown to be fork-free and prime. The case $n=4$ is shown in Fig.~\ref{fig:badcase}.
The graph $\Gv$ for the vertex labelled $v$ is also shown. This graph is clearly neither claw-free nor prime.
Consequently, Theorem 3 of~\cite{LozMil} is inapplicable, even in this simple case.
\begin{figure}[tbh]
\centering{
\begin{tikzpicture}[line width=0.5pt,minimum size=1.5pt,xscale=1,yscale=1]
\draw (-90:2) node[w] (0) {}  (-45:2) node[b] (1) {}  (0:2) node[w] (2) {} (45:2) node[b] (3) {}
(90:2) node[w,label=above:$v$] (4) {}  (135:2) node[b] (5) {}  (180:2) node[w] (6) {}  (225:2) node[b] (7) {} ;
\draw (0)--(1)--(2)--(3)--(4)--(5)--(6)--(7)--(0);
\draw (0)--(5) (2)--(7) (4)--(1) (3)--(6)  ;
\end{tikzpicture}\hspace*{3cm}
\begin{tikzpicture}[line width=0.5pt,minimum size=1.5pt,xscale=1,yscale=1]
\draw (-90:2) node[w] (0) {}  (0:2) node[w] (2) {}   (180:2) node[w] (6) {}  (225:2) node[b] (7) {} ;
\draw (6)--(7)--(0) (2)--(7) ;;
\end{tikzpicture}\vspace{5pt}
}
\caption{Graph $G$ and a derived $\Gv$}\label{fig:badcase}
\end{figure}
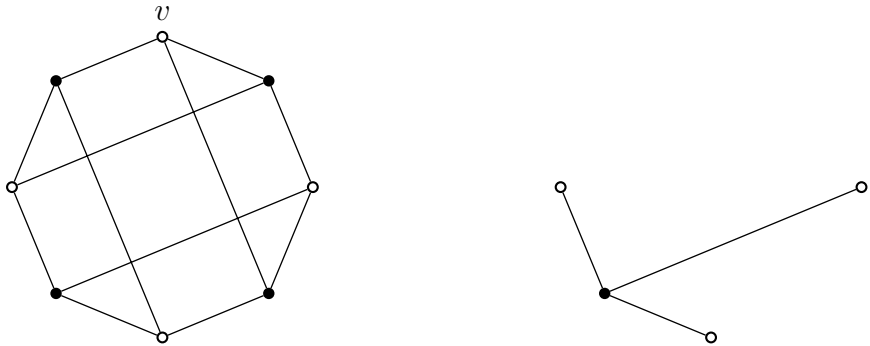

\subsection{Approximating $W(\Gv)$ for prime $G$ and $v\in V$}\label{sec:W(G_v)}\label{sec:Gv}

To apply Theorem~\ref{thm:prime}, we need the following strengthening.

\begin{corollary}\label{cor:extended}
Let $G$ be prime and $v$ be a vertex of $G$.  
The  modules $\tilde{M}_i$ ($i\in[h]$) in the extended decomposition tree for $\Gv$ are claw-free.
\end{corollary}
\begin{proof}
This follows directly from Observation~\ref{obs:isomorphic} and Theorem~\ref{thm:prime}.
\end{proof}
To determine $W(G)$ for a prime graph, we first show how to determine $W(\Gv)$ for any $v\in V$.
Let $G_v$ denote $\Gv$.  As we have seen, $G_v$ is not prime in general, so we must approximate $W(G_v)$ using the modular decomposition approach of Section~\ref{sec:moddecomp}.

The algorithm is then as follows. We construct the extended decomposition tree for $G_v$, with modules $\tilde{M_i}$ ($i\in[h]$). For each $i=1,2,\ldots,h$, we determine $W(G_{i-1}[\tilde{M_i}])$, using this as the weight for $v_{\tilde{M}_i}$ in $G_i$. From Corollary~\ref{cor:extended}, $G_{i-1}[\tilde{M_i}]$ is claw-free, so we may use the algorithm of Section~\ref{sec:clawoddholefree} in this computation. Finally $W(G_v)=w(u)$, where $u$ is the unique vertex in $G_h$.

More generally, suppose $G\in\cC$ for some hereditary class $\cC$ and, given $v\in V$, all prime subgraphs of $G_v$ are in some smaller hereditary class $\cA$. Then we can use this method to approximate $W(G_v)$ for graphs in $\cC$, using  modular decomposition and an algorithm for approximating $W$ for graphs in $\cA$. In our application $\cC=$ (fork,\,odd hole)-free and $\cA=$ (claw,\,odd hole)-free.

\subsection{Approximating $W(G)$ for prime $G$}\label{sec:partition}

The algorithm described in Section~\ref{sec:W(G_v)} approximates $W(\Gv)$ for prime $G$ and $v\in V$.
Let $S(v)=\{I\in \cI(G): v\in I\}$. Then $w(v)\,W(\Gv)=\sum_{I\in S(v)}w(I)$,
the total weight of all independent sets containing $v$. The classes $\cC$, $\cA$
are as in Section~\ref{sec:W(G_v)}.

We can write $V=\{v_1,v_2,\ldots,v_n\}$ and determine $w(v_i)\,W(\Gv[v_i])$ for $i\in[n]$,
similarly to~\cite{LozMil}, but the sum of these greatly overestimates $W(G)$, since
$\{S(v_i): i\in[n]\}$ is a cover of $\cI(G)$, not a partition.

Let $V_i=\{v_i,\dots,v_n\}$  and $S'(v_i)=\{I\in\cI(G): v_i\in I \text{~and~} I \subseteq V_i\}$. The sets $\{S'(v_i): i\in[n]\}$ form a partition of $\cI(G)\sm\{\es\}$, and so
\begin{equation}
W(G) = 1+\sum_{i=1}^n\sum_{I\in S'(v_i))}w(I)  = 1+\sum_{i=1}^n w(v_i)\,W(\GVN{i})\;.\label{eq:partition}
\end{equation}
So we must approximate $W(\GVN{i})$ for $i\in[n]$. We do this by constructing the extended decomposition tree for $\Gv[v_i]$ with leaf modules $\tilde{M}_1,\tilde{M}_2,\ldots,\tilde{M}_h$, as in Section~\ref{sec:W(G_v)}. From Corollary~\ref{cor:extended} we know that the modules $\tilde{M}_1,\tilde{M}_2,\ldots,\tilde{M}_h$ in this decomposition are in $\cA$. We transform this extended decomposition tree of $\Gv[v_i]$ into an extended decomposition tree for $\GVN{i}$.
For a fixed $i\in[n]$ we take $G'_0 = \GVN{i}$ and for $j\in[h]$ we set $G'_j = G_j\sm\{v_1,v_2,\dots,v_{i-1}\}$ and $\tilde{M}'_j=\tilde{M}_j\sm\{v_1,v_2,\dots,v_{i-1}\}$. For each $j\in[h]$ the set $\tilde{M}'_j$ is a module in $G'_{j-1}$ by Observation~\ref{obs:subgraph}.

To compute $W(\GVN{i})$, we note that restricting to $V_i\sm\nb[v_i]$ is equivalent to putting $w(v)=0$ for all $v \notin V_i\sm\nb[v_i]$. Thus we can use the algorithm of section~\ref{sec:Gv}, with exactly the same justification, after setting $w(v)=0$ for $v \notin V_i\sm\nb[v_i]$. Of course, in carrying out the algorithm we actually delete the vertices in $V\sm(V_i\sm\nb[v_i])$.
Thus the algorithm approximates $W(\GVN{i})$ for $i\in[n]$, using the algorithm of  Section~\ref{sec:clawoddholefree}, and then combines the estimates using~\eqref{eq:partition}.

\subsection{Approximating $W(G)$ for all graphs in $G\in\cC$}\label{sec:forkfree}

Since we can now approximate $W(G)$ for any prime $G$, we can
use the algorithm of Section~\ref{sec:modcount} to lift this to arbitrary $G$. This
completes the description of our algorithm.

The algorithm will fail if the (claw,\,odd hole)-free algorithm fails on any of the modules $\tilde{M}'_j$ for any prime $G$ and any $\GVN{i}$ $(i\in[h])$. In that case we conclude that $G$ is not (fork,\,odd hole)-free and terminate.

The basis of the algorithm is modular decomposition.
From Section~\ref{sec:modcount} this gives only a negligible overhead to the
algorithm for prime graphs in $\cC$.
From Section~\ref{sec:W(G_v)}, the algorithm for $\Gv$ is modular decomposition,
so this adds a negligible overhead. However, the algorithm for prime $G$
in Section~\ref{sec:partition} requires $n$ applications of the algorithm for $\Gv$.
Thus if $T_{\cA}(n,\epsilon)$ is the time complexity of the subroutine for $\cA$, the
overall time complexity is $O(nT_{\cA}(n,\epsilon n^{-2}))$.
(The error analysis is similar to the one seen earlier in the reduction to permanent approximation, and is given below.)
For $\cA=$ (claw,\,odd hole)-free graphs,
$T_{\cA}(n,\epsilon)=O(n^{12}\log^5(n)\epsilon^{-2})$, so for $\cC=$ (fork,\,odd hole)-free graphs,
$T_{\cC}(n,\epsilon)=O(n^{17}\log^5(n)\epsilon^{-2})$.

\subsubsection{Error Analysis}\label{sec:forkerror}

We can only approximate the total weight for graphs in $\cA$, in our application the class of
(claw,\,odd hole)-free graphs. So we must show that the resulting error in $W(G)$ can be
controlled for all $G\in\cC$, in our application the class of (fork,\,odd hole)-free graphs.

Suppose we approximate to a factor $(1\pm\varepsilon/2n^2)$ for $G\in\cA$.
In step $i$ of the algorithm of Section~\ref{sec:modcount}, we approximate $W(\tilde{M}_i)$,
and contract $\tilde{M}_i$ with this weight. Thus one weight in $G_{i+1}$ has error
$(1\pm\varepsilon/2n^2)$. We do this at most $n$ times, so the error in $W(G)$ becomes at most
$(1\pm\varepsilon/2n)$. We do this $n$ times in the method of Section~\ref{sec:partition},
and add the estimates. However, this does not increase the relative error. Finally,
we apply the algorithm of Section~\ref{sec:modcount} again, so the error is at most
$(1\pm\varepsilon/2n)^n$, which is at most $(1\pm\varepsilon)$ for
$\varepsilon<1$.

\subsection{Approximating $W_\alpha(G)$}\label{sec:fork W_k}
We cannot use the method of Section~\ref{sec:claw W_k} to approximate $W_k(G)$ for arbitrary $0\leq k\leq\alpha(G)$,
because there is no known analogue for fork-free graphs of the real-rootedness
result of~\cite{ChuSey} for claw-free graphs.
However, the method given in Section~\ref{sec:claw W_k} for estimating $W_\alpha(G)$ is valid for any graph class,
not necessarily even hereditary, where we can use arbitrary vertex weights. Therefore, the result of \cite{JSV}
for approximating the permanent can be completely generalised to approximating the total weight of maximum independent sets in (fork,\,odd hole)-free graphs.

\end{document}